\documentclass[pra,aps,preprintnumbers,superscriptaddress,nofootinbib]{revtex4}

\usepackage{hyperref}
\usepackage{verbatim}
\usepackage{amsmath}
\usepackage{latexsym}
\usepackage{revsymb}
\usepackage{yfonts}
\usepackage{color}
\usepackage{ifthen}
\usepackage{multirow}

\usepackage{natbib}
\usepackage{amsfonts}
\usepackage{amsmath}
\usepackage{amssymb}
\usepackage{amsthm}
\usepackage{graphicx}
\usepackage{bm}
\usepackage{bbm}

\newcommand{\be}{\begin{eqnarray} \begin{aligned}}
\newcommand{\ee}{\end{aligned} \end{eqnarray} }
\newcommand{\benn}{\begin{eqnarray*} \begin{aligned}}
\newcommand{\eenn}{\end{aligned} \end{eqnarray*} }

\newcommand{\sym}{ { \rm sym } } 
\newcommand{\bc}{\begin{center}}
\newcommand{\ec}{\end{center}}

\newcommand{\id}{\mathbb{I}}

\newcommand{\tr}{\mathop{\mathrm{tr}}\nolimits}

\newtheorem{theorem}{Theorem}[section]

\newtheorem{lemma}[theorem]{Lemma}
\newtheorem{definition}[theorem]{Definition}

\newtheorem{corollary}[theorem]{Corollary}
\newtheorem{proposition}[theorem]{Proposition}

\newcommand{\hil}{\mathcal{H}}
\newcommand{\hin}{\mathcal{H}_{\rm in}}
\newcommand{\hout}{\mathcal{H}_{\rm out}}
\newcommand{\msg}{M}
\newcommand{\epsball}{\mathcal{B}^{\eps}}

\newcommand{\states}{\mathcal{S}}
\newcommand{\substates}{\mathcal{S}_{\leq}}

\newcommand{\iSet}{\mathcal{I}}

\usepackage{amsfonts}
\def\Real{\mathbb{R}}
\def\Complex{\mathbb{C}}

\def\id{\mathbb{I}}

\def\01{\{0,1\}}

\newcommand{\eps}{\varepsilon}
\newcommand{\ket}[1]{|#1\rangle}
\newcommand{\bra}[1]{\langle#1|}

\newcommand{\proj}[1]{|#1\rangle\langle#1|}

\newcommand{\inp}[2]{\langle{#1}|{#2}\rangle} 

\newcommand{\hmin}{{H}_{\min}}
\newcommand{\hmineps}{{ H}_{\rm min}^{\eps}}

\newcommand{\bop}{\mathcal{B}}
\newcommand{\meas}{{\mathcal{M}}}

\newcommand{\nc}{\newcommand}

\nc{\ox}{\otimes}
\nc{\dn}{\downarrow}
\nc{\cA}{{\cal A}}
\nc{\cB}{{\cal B}}
\nc{\cC}{{\cal C}}
\nc{\cD}{{\cal D}}
\nc{\cE}{{\mathcal E}}
\nc{\cF}{{\cal F}}
\nc{\cG}{{\cal G}}
\nc{\cH}{{\cal H}}
\nc{\cI}{{\cal I}}
\nc{\cJ}{{\cal J}}
\nc{\cK}{{\cal K}}
\nc{\cL}{{\cal L}}
\nc{\cQ}{{\cal Q}}
\nc{\cM}{{\cal M}}
\nc{\cN}{{\cal N}}
\nc{\cO}{{\cal O}}
\nc{\cP}{{\cal P}}
\nc{\cR}{{\cal R}}
\nc{\cS}{{\cal S}}
\nc{\cT}{{\cal T}}
\nc{\cU}{{\cal U}}
\nc{\cV}{{\cal V}}
\nc{\cX}{{\cal X}}
\nc{\cZ}{{\cal Z}}

\nc{\isom}{\simeq}

\nc{\1}{{\openone}}

\newcommand{\ex}[1]	{\mathbf{E}\left\{ #1 \right\}}
\newcommand{\exc}[2]	{\mathbf{E}_{#1}\left\{ #2 \right\}}
\newcommand{\prob}[1]	{\mathbf{Pr}\left\{ #1 \right\}}
\newcommand{\probc}[2]	{\mathbf{Pr}_{#1}\left\{ #2 \right\}}

\nc{\entHmin}{\hmin}

\nc{\CC}{\mathbb{C}}

\newcounter{protoCount}
\newcounter{protoList}
\newsavebox{\tmpbox}
\newlength{\protobox}
\newenvironment{protocol}[3]{
\bigskip
\addtocounter{protoCount}{1}
\noindent \begin{lrbox}{\tmpbox}
\setlength{\protobox}{\textwidth}
\addtolength{\protobox}{-0.4cm}
\begin{minipage}[c]{\protobox}
\begin{bfseries}Protocol #1: #2\end{bfseries}
\ifthenelse{\equal{#3}{\empty}}{}{\\ #3}
\begin{list}{\begin{bfseries}\arabic{protoList}:\end{bfseries}}
{\usecounter{protoList}}
}{
\end{list}
\end{minipage}\end{lrbox}
\fbox{\usebox{\tmpbox}}
\bigskip
}
\newcommand{\Ext}{ {\rm Ext} }

\begin{document}
\title{Quantum to Classical Randomness Extractors}

\author{Mario \surname{Berta}}
\email[]{berta@phys.ethz.ch}
\affiliation{Institute for Theoretical Physics, ETH Zurich, 8093 Zurich, Switzerland.}
\author{Omar \surname{Fawzi}}
\email[ ] {ofawzi@cs.mcgill.ca}
\affiliation{School of Computer Science, McGill University, Montr\'eal, Qu\'ebec, Canada}
\author{Stephanie \surname{Wehner}}
\email[]{wehner@nus.edu.sg}
\affiliation{Centre for Quantum Technologies, National University of Singapore, 2 Science Drive 3, 117543 Singapore}

\date{\today}

\begin{abstract}
The goal of randomness extraction is to distill (almost) perfect randomness from a weak source of randomness. When the source yields a classical string $X$, many extractor constructions are known. Yet, when considering a physical randomness source, $X$ is itself ultimately the result of a measurement on an underlying quantum system. When characterizing the power of a source to supply randomness it is hence a natural question to ask, how much \emph{classical} randomness we can extract from a \emph{quantum} system. To tackle this question we here take on the study of \emph{quantum-to-classical randomness extractors} (QC-extractors).
	
	\begin{itemize}
		\item We provide constructions of QC-extractors based on measurements in a full set of mutually unbiased bases (MUBs), and certain single qubit measurements. The latter are particularly appealing since they are not only easy to implement, but appear throughout quantum cryptography. We proceed to prove an upper bound on the maximum amount of randomness that we could hope to extract from any quantum state. Some of our QC-extractors almost match this bound. We show two applications of our results.
		\item First, we show that any QC-extractor gives rise to entropic uncertainty relations with respect to quantum side information. Such relations were previously only known for two measurements. In particular, we obtain strong relations in terms of the von Neumann (Shannon) entropy as well as the min-entropy for measurements in (almost) unitary $2$-designs, a full set of MUBs, and single qubit measurements in three MUBs each.  
		\item Second, we finally resolve the central open question in the noisy-storage model [Wehner et al., PRL 100, 220502 (2008)] by linking security to the quantum capacity of the adversary's storage device. 	More precisely, we show that any two party cryptographic primitive can be implemented securely as long as the adversary's storage device has sufficiently low quantum capacity. Our protocol does not need any quantum storage to implement, and is technologically feasible using present-day technology.
	\end{itemize}
	
\vskip 1cm
\textbf{Keywords: } randomness extractors, randomness expansion, entropic uncertainty relations, mutually unbiased bases, quantum side information, two-party quantum cryptography, noisy-storage model.


\end{abstract}

\maketitle
\tableofcontents


\section{Introduction}

Randomness is an essential resource for information theory, cryptography, and computation. However, most sources of randomness exhibit only weak forms of unpredictability. The goal of randomness extraction is to convert such weak randomness into (almost) uniform random bits. Classically, a weakly random source simply outputs a string $X$ where the \lq amount\rq~of randomness is measured in terms of the probability of guessing the value of $X$ ahead of time. That is, it is measured in terms of the min-entropy $\hmin(X) = - \log P_{\rm guess}(X)$. To convert $X$ to perfect randomness, one applies a function $\Ext$ that takes $X$, together with a shorter string $R$ of perfect randomness (the \emph{seed}) to an output string $K = \Ext(X,R)$. The use of a seed is thereby necessary to ensure that the extractor works for all sources $X$ about which we know only the min-entropy, but no additional details of the source~\cite{vadhan:survey}.\footnote{Such as for example that each bit of a string $X$ is chosen independently.} Much work has been invested into showing that particular classes of functions have the property that $K$ is indeed very close to uniform as long as the min-entropy of the source $\hmin(X)$ is large enough. 

Yet, for most applications this is not quite enough, and we want an even stronger statement. In particular, imagine that we hold some \emph{side information} $E$ about $X$ that increases our guessing probability to $P_{\rm guess}(X|E)$. For example, such side information could come from an earlier application of an extractor to the same source. Intuitively, one would not talk about randomness if e.g., the output is uniformly distributed, but identical to an earlier output. In a cryptographic setting, side information can also be gathered by an adversary during the course of the protocol. We thus ask that the output is perfectly random even with respect to such side information, i.e., uniform \emph{and} uncorrelated from $E$. Classically, it is known that extractors are indeed robust against classical side information~\cite{robert:ext}, yielding a uniform output $K$, whenever the min-entropy about $X$ \emph{given access to side information} $E$ ($\hmin(X|E) = - \log P_{\rm guess}(X|E)$) is sufficiently high (see~\cite{Sha02, vadhan:survey} for surveys). Especially with respect to cryptographic applications, we thereby again want extractors that work for any source $X$ of sufficiently high entropy $\hmin(X|E)$ without any additional assumptions about the source.

Recently, it has been recognized that since the underlying world is not classical, $E$ may in fact hold \emph{quantum} side information about $X$~\cite{KMR05, RK05}. That this adds substantial difficulty to the problem was emphasized in~\cite{gavinsky:ext} where it was shown that 
there are in fact situations where using the same extractor gives a uniform output $K$ if $E$ is classical, but is entirely predictable when $E$ is quantum. Positive results were obtained in~\cite{RK05,robert:ext,renato:diss,TSSR10}, eventually culminating in~\cite{TS09, DPVR09}, proving that a wide class of classical extractors (with relatively short seed) yield a uniform output, as long as $\hmin(X|E)$ is sufficiently large.

Yet, in a fully quantum world we might ask ourselves: where does $X$ itself come from? How can we hope to harness even weak sources to obtain \emph{surplus} of classical randomness? Indeed, for any physical source hoping to create fresh randomness, $X$ is the result of a measurement on a quantum system $A$. That is, we can view the source as consisting of in fact two processes: First, a \emph{quantum} source emits a state $\rho_A$. Second, a measurement takes places yielding the classical string $X$. Note that quantum mechanics does allow many different measurements on $\rho_A$, and hence the question arises whether all such measurements are equally powerful at yielding a weakly random classical string $X$, or whether some are more useful to us than others. As such, it becomes clear that when trying to study our ability to extract randomness from any physical source, it is natural to ask how much randomness we can obtain from $\rho_A$ itself, rather than a particular classical string $X$.

The problem of extracting randomness from $X$ alone is further complicated by the fact that it is typically very hard to bound $\hmin(X|E)$, when $X$ is the result of quantum measurements on $A$, \emph{even} if we know stringent bounds on the \emph{quantum} correlations between $A$ and $E$ to begin with. When $E$ is trivial, entropic uncertainty relations~\cite{ww:URsurvey} yield such bounds when we are willing to average over a few randomly chosen measurements. A crude bound on $\hmin(X|E)$ can then be obtained by assuming that the size of $E$ is limited. But even classically, it is easy to see that there exist scenarios where bounding the adversaries' knowledge simply by his memory size yields very weak bounds~\cite{noisy:new}. Another approach to bounding $\hmin(X|E)$, common in e.g., Quantum Key Distribution (QKD), is possible in the case when randomness is extracted from a state $\rho_{ABE}$ where measurements are made on both $A$ and $B$ to obtain an estimate of $\hmin(X|E)$ where $X$ is obtained from $A$ alone~\cite{Tomamichel11,Tomamichel11_2,roger:thesis,pironio:nature,massar:recent,Colbeck11}. Part of the state is thereby consumed during the estimation process, which itself requires randomness. It is nevertheless possible to have an overall gain in randomness. For example, it is known that if measurements\footnote{That satisfy the no-signalling condition.} between systems $A$ and $B$ lead to a so-called Bell inequality violation, then $E$ knows little about $X$~\cite{roger:thesis,pironio:nature,massar:recent,Colbeck11,VV11,FGS11,PM11}. This is exactly the setting of the recent proofs~\cite{VV11,FGS11,PM11} of~\cite{Colbeck11,pironio:nature} where such violations were used to certify the creation of random bits using quantum measurements as a black box. 
Clearly, making such an estimate is only possible in a special setting where the states have a particular form $\rho_{ABE}$, and we are given access to $B$ \emph{and} $A$.


\subsection{Quantum to classical extractors}

This leads us to study \emph{quantum-to-classical randomness extractors} (QC-extractors). Our goal is to answer the following question: how can we extract \emph{classical} randomness from a physical source $\rho_{AE}$ by performing measurements on the \emph{quantum} state $\rho_A$? In analogy to classical extractors, we thereby want to obtain randomness from the source given only a minimal guarantee about its randomness - i.e.~like min-entropy $\hmin(X|E)$ for classical sources. It is important to note that unlike the classical world, quantum mechanics does allow for the creation of true randomness \emph{if} we are given full control of the source and can prepare any state $\rho_A$ at will.\footnote{For example, we could prepare the state $\ket{+} = (\ket{0}+\ket{1})/\sqrt{2}$ and measure it in the computational basis, yielding a truly random coin. Yet, this would correspond to controlling and knowing details of the source.} However, we want our extractors to work for \emph{any} unknown source as long as it has sufficiently high entropy.

As opposed to classical-to-classical extractors (CC-extractors) given by functions $\Ext(\cdot,R)$ mapping the outcome of the randomness source to a string $K$, a QC-extractor is described by projective measurements whose outcomes correspond to a classical string $K$. That is, a QC-extractor is a set of measurements $\left\{\meas^1_{A\rightarrow K},\ldots,\meas^L_{A \rightarrow K}\right\}$, where the random seed $R$ determines the measurement $\meas^R_{A\rightarrow K}$ that we will perform (see Section~\ref{sec:QCFromGeneralDecoupling} for a detailed explanation and a formal definition).\footnote{For quantum information theorists, note that one can of course use measurements to prepare states by measuring successively - however, recall that we are interested in how much randomness we can obtain from an \emph{unknown} source using a single measurement. The latter is furthermore motivated by experimental situations where successive measurements are typically very hard to implement.}

When talking about quantum states $\rho_{AE}$, what is the relevant measure of how weak or strong a source is? To gain some intuition on what the relevant measure should be, consider the case where $\rho_{AE}$ is the maximally entangled state between $A$ and $E$. Intuitively, this is the strongest quantum correlation that can exist between two systems. It is not hard to see that if we measure $A$ in \emph{any} basis to obtain some outcome $X$, and later communicate the choice of basis to an adversary holding $E$, then the adversary can guess $X$ perfectly. Intuitively, we would thus expect that the relevant measure of how weak a quantum source is with respect to $E$ involves a measure of the amount of entanglement
between $A$ and $E$. It turns out that the conditional min-entropy $\hmin(A|E)$ is exactly such a measure~\cite{krs:operational}, and we find that it is indeed the quantity that determines how many classical random bits we can hope to extract from $A$. That this is rather analogous to the classical case is very appealing. However, unlike for classical $A$, $\hmin(A|E)$ can be
\emph{negative} if $A$ is quantum (see below).

Note that in a quantum setting, we could also consider a quantum-to-quantum extractor (QQ-extractor). That is, an extractor in which we do not measure but merely ask that the resulting state is quantumly fully random (i.e.,~maximally mixed) and uncorrelated from $E$. Clearly, any QQ-extractor also forms a QC-extractor since any subsequent measurement on the maximally mixed state has a uniform distribution over outcomes. As such a QQ-extractor is stronger than a QC-extractor since we only require the output state to be close to uniform \emph{after} performing a measurement.\footnote{In quantum mechanics, it is possible to obtain a uniform distribution over outcomes even if the state was not maximally mixed. E.g., consider measuring the pure state $\proj{0}$ in the Fourier basis.} Constructions for such extractors are indeed well known in quantum information theory as a consequence of a notion known as \lq decoupling\rq, which plays a central role in quantum information theory (see~\cite{Horodecki05,Horodecki06,patrick:decouple,frederic:decoupling,Wullschleger08,Abeyesinghe06} and references therein). In general, a map that transforms a state $\rho_{AE}$ into a state that is close to a product state $\sigma_A \otimes \rho_E$ is a decoupling map. Decoupling processes thereby typically take the form of choosing a random unitary from a set $\left\{U_1,\ldots,U_L\right\}$ to $A=A_1A_2$ and tracing out (i.e., ignoring) the system $A_2$. For certain classes of unitaries such as (almost) unitary 2-designs~\cite{frederic:decoupling,patrick:blackHole,szehr:designs,Szehr11} (see below) the resulting state $\rho_{A_1E}$ is close to maximally mixed on $A_1$ and uncorrelated from $E$, whenever $\hmin(A|E)$ is sufficiently large. Measurements consisting of applying such a unitary, followed by a measurement on $A_1$ thus also yield QC-extractors.\footnote{For decoupling experts, note that the measurement map in a QC-extractor can be understood as a decoupling map. We would like to emphasize though that our QC-extractor results do not follow from previous work on decoupling, and our measurements have many nice properties not shared by unitaries used previously for decoupling.} Another example of QQ-extractors are given by protocols that aim to distill entanglement between $A$ and $B$ from a state $\rho_{ABE}$ by means of arbitrary communication between $A$ and $B$. The resulting output state is uncorrelated from $E$ and maximally mixed on (part of) $A$. The state has the additional requirement that when measuring on (part of) $A$ and $B$, the resulting output bits are perfectly correlated (i.e., they form a shared key). States $\rho_{AB}$ for which such a distillation is possible are also called \emph{private bits}~\cite{jonathan:pbits,Devetak04}. Note that given any QQ-extractor one could always purify the output onto an additional system, say, $B$. Being mixed on $A$ then corresponds naturally to being maximally entangled across $A$ and $B$ underlining the close relation between randomness extraction, and entanglement distillation~\cite{andris:extract,jonathan:pbits}. Note, however, that we do not want to assume special cases where we have access to other systems $B$ in order to perform such a distillation.

The authors of~\cite{BTS10} also proposed a definition of quantum extractors that is indeed somewhat similar to a QQ-extractor, however without any side information $E$. Our definitions (see Section~\ref{sec:QCFromGeneralDecoupling}) impose two important requirements not present in~\cite[Definition 5.1]{BTS10}. Firstly, we require the output of the extractor to be unpredictable for any, possibly quantum, adversary with access to side information $E$ provided $\entHmin(A|E)$ is large enough. Secondly, we consider \emph{strong} extractors so that even given the seed $R$, the output of the extractor cannot be predicted. This allows us to employ our extractor for cryptographic purposes. It also means that the output $K$ \emph{together with} $R$ are jointly close to uniform, meaning that we have effectively created \emph{more} almost perfect randomness than we invested in the seed.

\begin{itemize}
	\item[] \textbf{QC-extractors.} 
		\begin{itemize}
			\item We give two novel constructions of QC-extractors.\footnote{That is, not following from results on QQ-extractors (i.e., from general decoupling theorems in quantum information theory).} The first one involves a full set of mutually unbiased bases (MUBs) and pair-wise independent permutations (Theorem \ref{thm:full-set-mub}). This construction is more appealing than unitary 2-designs because it is combinatorially much simpler to describe and computationally more efficient, while having the same output size. 
		\item Our second construction (Theorem \ref{thm:singleQuditExtract}) is composed of unitaries acting on single qudits followed by some measurements in the computational basis. We also refer to these as \emph{bitwise} QC-extractors.  An appealing feature of the measurements defined by these unitaries is that they can be implemented \emph{with current technology}. In addition to computational efficiency, the fact that the unitaries act on a single qubit is often a desirable property for the design of cryptographic protocols in which the creation of randomness is not the only requirement for security. 
			Our example application below (see also Section~\ref{sec:noisy}) illustrates this.
		\item Finally, we also prove in Proposition~\ref{lem:optimality} that the maximum amount of randomness one can hope to extract is roughly $n + \entHmin(A|E)$, where $n$ denotes the input size. This upper bound can indeed be almost achieved by means of, e.g., our full set of MUBs QC-extractor. We also establish basic upper and lower bounds on the seed size for QC-extractors (see Table~\ref{tab:ext-summary}).
		\end{itemize}
\end{itemize}

The technique we use to prove that our constructions are QC-extractors is to bound the distance between the output of the extractor and the desired output in Hilbert-Schmidt norm (using ideas from~\cite{Horodecki05,Horodecki06,frederic:decoupling,Wullschleger08,szehr:designs,Szehr11,Berta08}). For the full set of MUBs, this distance can even be computed \emph{exactly}. We use the fact that the set of all the MUB vectors forms a complex projective 2-design and that the set of permutations is pair-wise independent. For our second construction, the analysis uses similar ideas in a more involved calculation. Our upper bound on the amount of extractable randomness follows from simple monotonicity properties of the min-entropy. The upper bound on the seed size follows from a non-explicit construction involving measure concentration techniques.


\subsection{Application to entropic uncertainty relations}

One of the fundamental ideas in quantum mechanics is the uncertainty principle. The security of essentially all quantum cryptographic protocols is founded on its existence. Intuitively, it states that even with complete knowledge about the quantum state $\rho_{A}$ of a system $A$, it is impossible to predict the outcomes of all possible measurements on $A$ with certainty. In an information theoretic context it is very natural to quantify this lack of knowledge in terms of entropic uncertainty relations (see~\cite{ww:URsurvey} for a survey). Apart from their deep significance in the foundations of quantum mechanics, entropic uncertainty relations are crucial tools in quantum information theory and quantum cryptography. The most well-known relation is for two measurements $\cM^{1}_{A\rightarrow K},\cM^{2}_{A\rightarrow K}$ and reads~\cite{Maassen88}
\begin{align}\label{eq:maassen}
\frac{1}{2}\sum_{j=1}^{2}H(K)_{\rho^{j}}\geq\log\frac{1}{c}\ ,
\end{align}
where $H(K)_{\rho^{j}}$ denotes the Shannon entropy of the post-measurement probability distributions $\rho^{j}_{K}=\cM^{j}_{A\rightarrow K}(\rho_{A})$, and $c$ measures the overlap between the measurements. Note that for any quantum state $\rho_{A}$ and measurements for which $c\neq1$, at least one of the entropies has to be greater than zero. In other words, it is impossible to predict the outcomes of both measurements with certainty. Uncertainty relations are thereby called \emph{strong}, if $\log(1/c)$ is large.

Just as extractors can depend on side information $E$, it is important to realize that also uncertainty should in fact not be treated as an absolute, but with respect to the prior knowledge of an observer who has access to a quantum system $E$~\cite{Winter10}. As an illustration, recall the example from above where $\rho_{AE}$ is the maximally entangled state. In this case, for any measurement on $A$, there is a corresponding measurement on $E$ that reproduces the measurement outcomes. I.e., there is no uncertainty at all! In order to take into account possibly quantum information about $A$, one needs to prove new entropic uncertainty relations that would have an additional term quantifying the quantum side information. Unfortunately, up to this day, we only know such relations for \emph{two} measurements~\cite{Berta09,Joe09,Coles11,Coles12,Coles10,christandl05,Tomamichel11}. Intuitively, uncertainty relations for two measurements are much easier to prove than relations for more measurements as in this case uncertainty coincides with another foundational notion in quantum information, complementarity. This notion is relevant when we perform two measurements in succession and was an essential ingredient in the proofs. However, it does not carry over to three or more measurements. Here, we prove the following results.

\begin{itemize}
	\item[] \textbf{Uncertainty relations with quantum side information for more than two measurements.} We show that any set of measurements forming a QC-extractor yields an entropic uncertainty relation \emph{with} respect to quantum side information. We thereby obtain relations both for the usual von Neumann (Shannon) entropy, as well as the min-entropy. The latter is relevant for cryptographic applications. This yields the first uncertainty relations with quantum side information for more than two measurements. From our QC-extractors, we obtain strong uncertainty relations for (almost) unitary 2-designs, measurements in a full set of mutually unbiased bases (MUBs) on the whole space, as well as on many single qudits. The latter are the measurements used e.g.,~in the six-state protocol of QKD, and are particularly relevant for applications in quantum cryptography (see Table~\ref{tab:URsummary} for a summary of results for the min-entropy).
\end{itemize}

Note that uncertainty relations in terms of the min-entropy effectively help us to bound $\hmin(X|ER)$, where $R$ is the seed for the QC-extractor (see Section~\ref{sec:URbounds} for details).
For example, for the full set of MUBs we prove that 
\begin{align}
	\hmin(X|ER) \gtrsim \log|A| + \hmin(A|E)\ ,
\end{align}
where the output of the measurements is called $X$. Since $\hmin(A|E)$ is negative when $A$ and $E$ are entangled, one obtains less uncertainty in this case (as expected when considering the example of a maximally entangled state given above). Of course, given such a bound, we could in turn apply a CC-extractor to the weakly random string $X$ to obtain a uniform $K$. This underscores the beautiful relation between the concept of randomness extraction from a quantum state, and the notion of uncertainty relations with side information in quantum physics. From a QC-extractor, we obtain uncertainty relations. In turn, from any measurements inducing strong uncertainty relations \emph{plus} a CC-extractor, we obtain a QC-extractor.\footnote{Note that measurements plus a classical post-processing effectively forms a new, larger, set of measurements.}


\subsection{Application to cryptography}

Our second application is to proving security in the noisy-storage model. Unfortunately, it turns out that even quantum communication does not enable us to solve two-party cryptographic problems between two parties that do not trust each other~\cite{lo:insecurity}. Such problems include e.g., the well-known primitives bit commitment and oblivious transfer~\cite{lo&chau:bitcom, 
lo:promise,mayers:bitcom,wehner06d,kretch:bc}, of which merely very weak variants are possible. How can this be when quantum communication offers such great advantages when it comes to distributing encryption keys? Intuitively, the security proof of QKD is considerably simplified by the fact that Alice and Bob do trust each other, and can collaborate to check for any eavesdropping activity. For example, as mentioned above, when Alice and Bob share a state $\rho_{ABE}$, where the eavesdropper holds $E$, they can use up part of the state to obtain an estimate of $\hmin(X|E)$, where $X$ is a measurement outcome of the remaining part of Alice's system.

Yet, since two-party cryptographic protocols are a central part of modern cryptography, one is willing to make assumptions on how powerful the adversary can be in order to obtain security. Classically, these assumptions typically consist of two parts. First, one assumes that a particular problem requires a lot of computational resources to solve in some precise complexity theoretic sense. Second, one assumes that the adversary does indeed have insufficient computational resources. However, we might instead ask whether there are other, more \emph{physical} assumptions that enable us to solve such tasks?

Classically, it is possible to obtain security, when we are willing to assume that the adversary's \emph{classical} memory is limited in size~\cite{Maurer92b,cachin:bounded}. Yet, apart from the fact that classical storage is by now cheap and plentiful, the beautiful idea of assuming a limited classical storage has one rather crucial caveat: \emph{any} classical protocol in which the honest players need to store $n$ classical bits to execute the protocol can be broken by an adversary who is able to store more than $O(n^2)$ bits~\cite{maurer:imposs}. Motivated by this unsatisfactory gap, it was thus suggested to assume that the attacker's \emph{quantum} storage was bounded~\cite{serge:bounded,serge:new}, or, more generally, noisy~\cite{Noisy1, noisy:robust,noisy:new}. The central assumption of the so-called \emph{noisy-storage model} is that during waiting times $\Delta t$ introduced in the protocol, the adversary can only keep quantum information in his quantum storage device $\cF$. Otherwise, the attacker may be all powerful. In particular, he can store an unlimited amount of classical information, and perform computations \lq instantaneously\rq. The latter implies that the attacker could encode his quantum information into an arbitrarily complicated error correcting code to protect it from any noise in $\cF$ (see Section~\ref{sec:noisy} for details). Of particular interest are thereby quantum memories consisting of $N$ \lq memory cells\rq, each of which undergoes some noise described by a channel $\cN$. That is, the memory device is of the form $\cF = \cN^{\otimes N}$. Note that the bounded storage model is a special case, where each memory cell is just one qubit, and $\cN$ is the identity channel. To relate the number of transmitted qubits $n$ to the size of the storage device one typically chooses the \emph{storage rate} $\nu$ such that $N = \nu \cdot n$. We follow this convention here to ease comparison with earlier work.

Since its inception~\cite{Noisy1}, it was clear that security in the noisy-storage model should be related to the question of how much information the adversary can send through his noisy storage device. That is, the capacity of $\cF$ to transmit quantum information. Initial progress was made in~\cite{noisy:new} where security was linked to the storage device's ability to transmit \emph{classical} information and shown against fully general 
attacks.\footnote{Before~\cite{noisy:new}, security was only shown under the additional assumption that the adversary attacks each qubit individually~\cite{Noisy1}. Whereas this may sound similar to problems in QKD, note that the setting is entirely different when proving security between two mutually distrustful parties, and security in QKD does not imply security in this model.} Further progress was made only very recently, linking the security to the so-called entanglement cost of the storage device~\cite{entCost}, which lies between its classical and quantum capacities.

\begin{itemize}
	\item[] {\bf Security and the quantum capacity.} Here, we finally resolve the question of linking security in the noisy-storage model to the quantum capacity of the storage device. 
		More precisely, we show that
		any two-party cryptographic primitive can be implemented securely under the assumption that the adversary is restricted to using
		a quantum storage device of the form $\cF = \cN^{\otimes \nu \cdot n}$ by means of a protocol transmitting $n$ qubits whenever
		\begin{align}
			\nu \cdot \cQ(\cN) < 1\ ,\mbox{ and }
			2 - \log(3) \lesssim \nu \cdot \gamma^Q(\cN,1/\nu)\ ,
		\end{align}
		where $\cQ(\cN)$ is the quantum capacity of the channel $\cN$ and $\gamma^Q(\cN,1/\nu)$ is the so-called strong converse parameter of $\cN$ for sending information through $\cF$
		at rate $R = 1/\nu$. Note that the second condition actually \emph{does} favor small $\nu$, since $\gamma^Q(\cN,1/\nu)$ is large whenever the rate $R = 1/\nu$ is large.
		A similar statement can be obtained for general channels $\cF$ (see Section~\ref{sec:noisy} for details and a worked out example).
\end{itemize}

We prove our result by showing the security of a simple quantum protocol for the cryptographic primitive \emph{weak string erasure}~\cite{noisy:new},
which is known to be universal for two-party secure computation~\cite{noisy:new}. To this end, we employ the bitwise QC-extractor for measurements of single qubits, each in one of three MUBs, known from the six-state protocol in QKD. 


\section{Preliminaries}

\subsection{Basic concepts}

We briefly recount some important facts of quantum information, and establish notational conventions. A more gentle introduction can be found in e.g.~\cite{noisy:new} or~\cite{NieChu00Book}.

\subsubsection{Quantum states}

In quantum mechanics, a system such as Alice's or Bob's labs are described mathematically by \emph{Hilbert spaces}, denoted by $A, B, C,\ldots$.
Here, we follow the usual convention in quanutm cryptography and assume that all Hilbert spaces are finite-dimensional. We write $|A|$ for the dimension of $A$. The set of linear operators on $A$ is denoted by $\cL(\cA)$. A \emph{quantum state} $\rho_{A}$ is an operator $\rho_{A}\in\states(A)$, where $\states(A)=\{\sigma_{A}\in\cL(A)\mid\sigma_{A}\geq0,\tr(\sigma_{A})=1\}$. If $\rho_{A}$ has rank $1$ it is called a \emph{pure} state. For technical reasons we also need the notion of sub-normalized states $\rho_{A}\in\substates(A)$, where $\substates(A)=\{\sigma_{A}\in\cL(A)\mid\sigma_{A}\geq0,\tr(\sigma_{A})\leq1\}$. We will use the term \emph{state} to refer to sub-normalized states, unless otherwise indicated in context.

Two systems $A$ and $B$ are combined using the tensor product, written as $AB\equiv A\otimes B$. An operator on two systems $AB$ is thereby also called \emph{bipartite} (and \emph{multipartite} if the number of systems is larger). Given a bipartite state $\rho_{AB}\in\cS_{\leq}(AB)$, we write $\rho_{A}=\tr_{B}[\rho_{AB}]$ for the corresponding reduced state, where $\tr_B$ is the partial trace over $B$. That is, $\rho_A$ is the state on system $A$ alone. 

It will be convenient to express classical probability distributions as quantum states. For some set $\cX$, let $\{\ket{x}\}_{x\in\cX}$ be an orthonormal basis of the space $X$ where each basis vector $\ket{x}$ corresponds to some particular element $x \in \cX$. A distribution $P_X$ over $\cX$ can now be expressed as 
\begin{align}\label{eq:classicDist}
	\rho_X = \sum_{x \in \cX} P_X(x) \proj{x}\ .
\end{align}
We also call this a \emph{classical state} or a $c$-state. In general, systems are called classical if they are of the above form for some fixed standard basis, often called the computational basis.
Naturally, one can now also consider states which are classical on system $X$ and quantum on some other system $A$. Such states have the form
\begin{align}
	\rho_{XA} = \sum_{x \in \cX} P_X(x) \underbrace{\proj{x}}_{X} \otimes \underbrace{\rho_A^x}_{A}\ .
\end{align}
We also call such states \emph{classical-quantum} or $cq$-states. In general, when indicating that a multipartite state is part classical, part quantum we will use $c$ and $q$ to label the classical and quantum systems, respectively.


\subsubsection{Quantum operations}

The simplest quantum operation is given by a unitary operator $U$ taking $\rho$ to $U \rho U^\dagger$. Later on, we will consider applying unitary operators only to one part of a multipartite state. 
When applying $U$ only to system $A$ of $\rho_{AB}$ we thereby also use the common shorthand
\begin{align}
	U_A \rho_{AB} U_A^\dagger = (U\otimes \id_B) \rho_{AB} (U \otimes \id_B)^\dagger\ ,
\end{align}
where $\id_{B}$ denotes the identity in $\cL(B)$. More generally, for $M_{A}\in\cL(A)$, we write $M_{A}\equiv M_{A}\otimes\id_{B}$ for the enlargement on any $AB$.  Any operation allowed by quantum mechanics can be expressed as a quantum channel. The simplest of these is the identity channel. For $A$, $B$ with orthonormal bases $\{\ket{i}_{A}\}_{i=1}^{|A|}$, $\{\ket{i}_{B}\}_{i=1}^{|B|}$ and $|A|=|B|$, the canonical identity mapping from $\cL(A)$ to $\cL(B)$ with respect to these bases is denoted by $\cI_{A\rightarrow B}$, i.e.~$\cI_{A\rightarrow B}(\ket{i}\bra{j}_{A})=\ket{i}\bra{j}_{B}$. A linear map $\cE_{A\rightarrow B}:\cL(A)\rightarrow\cL(B)$ is positive if $\cE_{A\rightarrow B}(\rho_{A})\geq0$ for all $\rho_{A}\geq0$. It is completely positive if the map $(\cE_{A\rightarrow B}\otimes\cI_{C\rightarrow C})$ is positive for all $C$. Completely positive and trace preserving maps (CPTMs) are called \emph{quantum channels}.

Indeed, also a measurement can be described as a quantum channel. Intuitively, a measurement takes a state $\rho$ to one of several possible classical measurement \lq outcomes\rq, where each outcome occurs with a certain probability. That is, for some fixed measurement a particular state $\rho$ determines some classical probability distribution over outcomes. Recall from Equation~\eqref{eq:classicDist} that we can express this distribution in terms of a quantum state. It will be convenient to express this in terms of a  quantum channel as the following \emph{measurement map}, that we will need in Section~\ref{sec:QCFromGeneralDecoupling}. For a bipartite system $A=A_{1}A_{2}$, it is defined as $\cT_{A\rightarrow A_1}:\cL(A)\to\cL(A_1)$,
\begin{align}
\label{eq:meas-map}
\cT(.)_{A\rightarrow A_1}= \sum_{a_{1}a_{2}} \bra{a_{1}a_{2}}(.)\ket{a_{1}a_{2}} \proj{a_1}\ ,
\end{align}
where $\{\ket{a_{1}}\},\{\ket{a_{2}}\}$ are (standard) orthonormal bases of $A_{1},A_{2}$ respectively.  A small calculation readily reveals that this map can be understood as tracing out $A_2$, and then measuring the remaining system $A_{1}$ in a basis $\{\ket{a_{1}}\}$. Note that the outcome of the measurement map is classical in the basis $\{\ket{a_{1}}\}$ on $A_{1}$.

Throughout, we will need this measurement map to consider measurements of a specific form. These are formed by first applying some particular unitary $U_{j}$ to the state, followed by the measurement map $\cT_{A\rightarrow A_{1}}$. We denote these measurements by
\begin{align}\label{eq:meas}
\cM_{A\rightarrow K_{1}}^{j}(\rho_{A})=\cI_{A_{1}\rightarrow K_{1}}\left(\cT_{A\rightarrow A_{1}}\left(U_j \rho_A U_j^\dagger\right)\right)\ ,
\end{align}
where the relabeling $A_{1}\rightarrow K_{1}$ accounts for the fact that the output system is actually classical (a notation that will be very useful in Section~\ref{sec:URbounds} on entropic uncertainty relations).
 

\subsubsection{Distance measures}
		
We will employ two well known distance measures between quantum states. The first is the $L_{1}$- or \emph{trace distance}, which is 		
induced by the $L_{1}$-norm $\|\rho\|_1 = \tr\left[\sqrt{\rho^\dagger \rho}\right]$. The trace distance determines the success probability of distinguishing two states $\rho$ and $\sigma$ given with a priori equal probability~\cite{helstrom}.

The second distance measure we will refer to is the \emph{purified distance}. To define it, we need the concept of \emph{generalized fidelity} between two states $\rho,\sigma$, which can be defined as~\cite{Tomamichel09},
\begin{align}
\bar{F}(\rho,\sigma)=F(\rho,\sigma)+\sqrt{\left(1-\tr[\rho]\right)\left(1-\tr[\sigma]\right)}\ ,
\end{align}
where $F(\rho,\sigma)=\|\sqrt{\rho}\sqrt{\sigma}\|_1$ is the usual notion of \emph{fidelity}. Note that if at least one of states is normalized, then the two notions of fidelity coincide, i.e.~$\bar{F}(\rho,\sigma)=F(\rho,\sigma)$. The purified distance between two states $\rho,\sigma$ is now defined as~\cite{gilchrist:distance,Tomamichel09}
\begin{align}
P(\rho,\sigma)=\sqrt{1-\bar{F}(\rho,\sigma)^2}\ ,
\end{align}
and is a metric on the set of sub-normalized states~\cite{Tomamichel09}. To gain some intuition about the notion of purified distance, note that by Uhlman's theorem~\cite{uhlmann} the fidelity between two normalized states $\rho,\sigma$ can be written as $F(\rho,\sigma)=\max_{\ket{\rho},\ket{\sigma}}|\inp{\rho}{\sigma}|$, where the maximization is taken over all purifications $\proj{\rho}$ of $\rho$ and $\proj{\sigma}$ of $\sigma$. Furthermore, note that for pure states $\sqrt{1 - F(\proj{\rho},\proj{\sigma})^2} = \frac{1}{2} \| \proj{\rho} - \proj{\sigma}\|_1$. Hence, for normalized states, we can think of the purified distance as the minimal trace distance between any two purifications of the states $\rho$ and $\sigma$. The purified distance is indeed closely related to the trace distance, as for any two states $\rho,\sigma$ we have~\cite{Tomamichel09},
\begin{align}\label{eq:purifiedVStrace}
\frac{1}{2} \| \rho - \sigma \|_1 \leq P(\rho,\sigma) \leq \sqrt{2 \| \rho - \sigma \|_1}\ .
\end{align}
It is furthermore easy to see that for normalized states the factor $2$ on the right hand side can be improved to $1$.

For any distance measure, we can define an $\eps$-ball of states around $\rho$ as the states at a distance not more than $\eps$ from $\rho$. Below, we will apply this notion to the purified distance and define
\begin{align}
\cB^{\eps}(\rho_{A})=\{\sigma_{A}\in\substates(A)\mid P(\rho_{A},\sigma_{A})\leq\eps\}\ .
\end{align}
		

\subsection{Quantifying information}

The \textit{von Neumann entropy} of $\rho_{A}\in\cS_{\leq}(A)$ is defined as $H(A)_{\rho}=-\tr[\rho_{A}\log\rho_{A}]$. 
Note that for a classical state $\rho_X$ this is simply the familiar Shannon entropy.
The \emph{conditional von Neumann entropy of $A$ given $B$} for $\rho_{AB}\in\cS_{\leq}(AB)$ is defined as
\begin{align}
H(A|B)_{\rho}=H(AB)_{\rho}-H(B)_{\rho}\ .
\end{align}
The \emph{conditional min-entropy} of a state $\rho_{AB}\in\states(AB)$ defined as\footnote{We write $\max$ instead of $\sup$ as we work with finite dimensional Hilbert spaces.}
\begin{align}
\entHmin(A|B)_{\rho}=\max_{\sigma_B \in \states(B)}\entHmin(A|B)_{\rho|\sigma}\ ,
\end{align}
with
\begin{align}		
\entHmin(A|B)_{\rho|\sigma}=\max\left\{\lambda \in \Real: 2^{-\lambda}\cdot\id_A \otimes \sigma_B \geq \rho_{AB}\right\}\ .
\end{align}
For the special case where $B$ is trivial, we obtain $\hmin(A)_{\rho}=-\log\|\rho_{A}\|_{\infty}$, where $\|.\|_{\infty}$ denotes the operator norm.

Whereas this definition may seem rather unwieldy, the min-entropy is known to have interesting operational interpretations~\cite{krs:operational}. If $A$ is classical, then the min-entropy can be expressed as $\entHmin(A|B)_{\rho}=-\log P_{\rm guess}(A|B)$, where $P_{\rm guess}(A|B)$ is the average probability of guessing the classical symbol $A=a$ maximized over all possible measurements on $B$. If $A$ is quantum, then $\entHmin(A|B)_{\rho}$ is directly related to the maximal achievable singlet fraction achievable by performing an operation on $B$, i.e.~it is intuitively related to the amount of entanglement between $A$ and $B$.

In practice, the full (operational) use of entropies only comes to play if one works with \emph{smoothed entropies}.\footnote{Of course, this is not the case for the von Neumann entropy. But note that the von Neumann entropy usually only has operational interpretations in an independent and identically distributed asymptotic setting. In contrast to this, smooth entropies allow the quantitative characterization of general (structureless) resources.} For the conditional min-entropy this takes the form
\begin{align}\label{eq:smoothMinDef}
\entHmin^{\eps}(A|B)_{\rho}=\max_{\tilde{\rho}_{AB} \in \epsball(\rho_{AB})} \entHmin(A|B)_{\tilde{\rho}}\ ,
\end{align}
where the smoothing parameter $\eps\geq0$ typically corresponds to an error tolerance in information theoretic operational interpretations. For a more detailed discussion about smooth entropies we refer to~\cite{renato:diss,krs:operational,Tomamichel08,Tomamichel09,datta-2008-2}.
		

\section{Quantum to Classical Randomness Extractors (QC-Extractors) }\label{sec:QCFromGeneralDecoupling}

The use of random bits is of fundamental importance for many information theoretic and computational tasks. However, perfect randomness is not easily found in nature. Most sources of randomness only exhibit weak forms of unpredictability. In order to use such sources in applications, one has to find a procedure to convert weak randomness into almost uniform random bits. Such procedures are usually referred to as randomness extractors, which have been extensively studied in the theoretical computer science literature; see \cite{Sha02, vadhan:survey} for surveys. 

In a classical world, the sources of randomness are described by probability distributions and the randomness extractors are families of (deterministic) functions taking each possible value of the source to a binary string. To understand the definition of quantum extractors, it is convenient to see a classical extractor as a family of permutations acting on the possible values of the source. This family of permutations should satisfy the following property: for any probability distribution on input bit strings with high min-entropy, applying a typical permutation from the family to the input induces an almost uniform probability distribution on a prefix of the output. We define a quantum to quantum extractor in a similar way by allowing the operations performed to be general unitary transformations and the input to the extractor to be quantum.

\begin{definition}[QQ-Extractors]\label{def:QQ-extractor}
Let $A = A_1 A_2$ with $n = \log|A|$. Define the trace-out map $\tr_{A_2} : \cL(A) \to \cL(A_1)$ by $\tr_{A_2}(.) = \sum_{a_{2}} \bra{a_{2}}(.)\ket{a_{2}}\ ,$ where $\{\ket{a_{2}}\}$ is an orthonormal basis of $A_{2}$.

For $k \in [- n , n ]$ and $\eps \in [0,1]$, a $(k, \eps)$-QQ-extractor is a set $\{U_1, \dots, U_{L}\}$ of unitary transformations on $A$ such that for all states $\rho_{AE} \in \cS(AE)$ satisfying $\entHmin(A|E)_{\rho} \geq k$, we have
\begin{equation}\label{eq:QQ-extractor}
\frac{1}{L} \sum_{i=1}^L \left\| \tr_{A_2}\left[U_i \rho_{AE} U_i^{\dagger}\right] - \frac{\id_{A_1}}{|A_1|} \ox \rho_E \right\|_1 \leq \eps \ .
\end{equation}
$\log L$ is called the seed size of the $QQ$-extractor.
\end{definition}
We make a few remarks on the definition. First, we should stress that the same set of unitaries should satisfy \eqref{eq:QQ-extractor} \emph{for all} states $\rho_{AE}$ that meet the conditional min-entropy criterion $\entHmin(A|E)_{\rho} \geq k$. In particular, the system $E$ can have arbitrarily large dimension. The quantity $\entHmin(A|E)_{\rho}$ measures the uncertainty that an adversary has about the system $A$. As it is usually impossible to model the knowledge of an adversary, a bound on the conditional min-entropy is often all one can get. A notable difference with the classical setting is that the conditional min-entropy $k$ can be negative when the systems $A$ and $E$ are entangled. In fact, in many cryptographic applications, this case is the most interesting.

A statement of the form of Equation~\eqref{eq:QQ-extractor} is more commonly known as a \lq decoupling\rq~result~\cite{Horodecki05, Horodecki06,frederic:decoupling, Abeyesinghe06, Wullschleger08,patrick:decouple}. Note, however, that decoupling does not always lead to the output being close to maximally mixed. Such statements play an important role in quantum information theory and many coding theorems amount to proving a decoupling theorem. In fact, the authors of~\cite{frederic:decoupling, Wullschleger08} showed that a set of unitaries forming a unitary 2-design (see Definition \ref{def:two-design}) define a $(k,\eps)$-QQ-extractors as long as the output size $\log |A_{1}|\leq(n+k)/2-\log(1/\eps)$.

A definition of quantum extractors was also proposed in~\cite[Definition 5.1]{BTS10}. Our definition is stronger in two respects. Firstly, we consider \emph{strong} extractors in that we impose a condition on the average of the trace distance to the uniform distribution by contrast to the trace distance of the average. The weaker constraint used by~\cite{BTS10} allows them to construct quantum extractors with output size equals to the input size.\footnote{In this case, the net randomness extracted is obtained by subtracting the randomness used for the seed} Secondly, we require the extractor to decouple the $A$ system from any \emph{quantum} side information held in the system $E$.

In the context of cryptography, a QQ-extractor is often more than one needs. In fact, it is usually sufficient to extract random \emph{classical} bits, which is in general easier to obtain than random qubits. This motivates the following definition, where the difference to a QQ-extractor is that the output system $A_1$ is measured in the computational basis. In particular, any $(k, \eps)$-QQ-extractor is also a $(k, \eps)$-QC-extractor. 

\begin{definition}[QC-Extractors]\label{def:qc-extractor}
Let $A = A_1 A_2$ with $n = \log |A|$, and let $\cT_{A\rightarrow A_{1}}$ be the measurement map defined in Equation~\eqref{eq:meas-map}.

For $k \in [-n, n]$ and $\eps \in [0,1]$, a $(k, \eps)$-QC-extractor is a set $\{U_1, \dots, U_{L}\}$ of unitary transformations on $A$ such that for all states $\rho_{AE} \in \cS(AE)$ satisfying $\entHmin(A|E) \geq k$, we have
\begin{align}\label{eq:qc-extractor}
\frac{1}{L} \sum_{i=1}^L \left\| \cT_{A\rightarrow A_{1}}(U_i \rho_{AE} U_i^{\dagger})- \frac{\id_{A_1}}{|A_1|} \ox \rho_E \right\|_1 \leq \eps\ .
\end{align}
$\log L$ is called the seed size of the $QC$-extractor.
\end{definition}

Observe that Definition~\ref{def:qc-extractor} only allows a specific form of measurements obtained by applying a unitary transformation followed by a measurement in the computational basis of $A_1$. The reason we use this definition is that we want the output of the extractor to be determined by the source and the choice of the seed. In the quantum setting, a natural way of translating this requirement is by imposing that an adversary holding a system that is maximally entangled with the source can perfectly predict the output. This condition is satisfied by the form of measurements dictated by Definition~\ref{def:qc-extractor}. Allowing generalized measurements (POVMs) already (implicitly) allows the use of randomness for free. Note also, that in the case where the system $E$ is trivial, a $(0, \eps)$-QC-extractor is the same as an $\eps$-metric uncertainty relation \cite{fawzi11}.


\subsection{Examples and limitations of QC-extractors}

Universal (or two-independent) hashing is probably one of the most important extractor constructions, which even predates the general definition of extractors \cite{ILL89}. Unitary 2-designs can be seen as a quantum generalization of two-independent hash functions.

\begin{definition}\label{def:two-design}
A set of unitaries $\{U_1, \dots, U_L\}$ acting on $A$ is said to be a 2-design if for all $M \in \cL(A)$, we have
\begin{align}
\frac{1}{L} \sum_{i=1}^L U_i^{\otimes 2} M (U_i^{\dagger})^{\otimes 2} = \int U^{\otimes 2} M (U^{\dagger})^{\otimes 2} dU
\end{align}
where the integration is with respect to the Haar measure on the unitary group.
\end{definition}

Many efficient constructions of unitary 2-designs are known~\cite{DCEL09, gross07}, and in an $n$-qubit space, such unitaries can typically be computed by circuits of size $O(n^2)$. However, observe that the number of unitaries of a 2-design is at least $L \geq|A|^4-2|A|^2+2$~\cite{gross07}. The following is immediate using a general decoupling result from~\cite{frederic:decoupling,Wullschleger08} (see Lemma~\ref{lem:2design}).

\begin{corollary}\label{thm:qc-ext-two-design}
Let $A = A_1 A_2$ with $n = \log |A|$. For all $k \in [-n, n]$ and all $\eps > 0$, a unitary 2-design $\{U_1, \dots, U_L\}$ on $A$ is a $(k, \eps)$-QC-extractor with output size
\begin{align}
\log |A_1| = \min(n, n+k-2\log(1/\eps)).
\end{align}
\end{corollary}

Similar results also hold for almost unitary 2-designs; see~\cite{szehr:designs,Szehr11}. Using~\cite{Harrow09_2}, this shows for instance that random quantum circuits of size $O(n^2)$ are QC-extractors with basically the same parameters as in Corollary~\ref{thm:qc-ext-two-design}. We now prove that choosing a reasonably small set of unitaries at random defines a QC-extractor with high probability. The seed size in this case is of the same order as the output size of the extractor. We expect that a much smaller seed size would be sufficient.

\begin{theorem}\label{thm:small-decoupling-set}
Let $A=A_{1}A_{2}$ with $n = \log |A|$ and $\cT_{A\rightarrow A_{1}}$ be the measurement map defined in Equation~\eqref{eq:meas-map}. Let $\eps > 0$, $c$ be a sufficiently large constant, and
\begin{align}
\log |A_1| \leq n + k - 4 \log (1/\eps) - c \qquad \text{ as well as } \qquad \log L \geq \log |A_1| + \log n + 4 \log(1/\eps) + c \ .
\end{align}
Then, choosing $L$ unitaries $\left\{U_1, \dots, U_L\right\}$ independently according to the Haar measure defines a $(k, \eps)$-QC-extractor with high probability.
\end{theorem}

The proof can be found in Appendix~\ref{app:proofs}. It uses one-shot decoupling techniques~\cite{Berta08,frederic:decoupling,Wullschleger08,Szehr11,szehr:designs} combined with an operator Chernoff bound~\cite{AW02} (see Lemma~\ref{lem:operator-chernoff}).

We now give some limitations on the output size and seed size of QC-extractors. The following lemma shows that even if we are looking for a QC-extractor that works for a particular state $\rho_{AE}$, the output size is at most $n+H_{\min}^{\sqrt{\eps}}(A|E)_{\rho}$, where $n$ denotes the size of the input.

\begin{proposition}[Upper bound on the output size]\label{lem:optimality}
	Let $A=A_{1}A_{2}$, $\rho_{AE}\in\cS(AE)$, $\{U_1, \dots, U_L\}$ a set of unitaries on $A$, and $\cT_{A\rightarrow A_{1}}$ defined as in Equation~\eqref{eq:meas-map}, such that
	\begin{align}
	\label{eq:opt-ext-cond}
		\frac{1}{L} \sum_{i=1}^L \left\| \cT_{A\rightarrow A_{1}}\left( U_i\rho_{AE} U_i^{\dagger} \right) - \frac{\id_{A_1}}{|A_1|} \otimes \rho_E\right\|_1 \leq \eps\ .
	\end{align}
	Then,
	\begin{align}
		\log |A_1| \leq \log |A| + \entHmin^{\sqrt{\eps}}(A|E) \ . 
	\end{align}
\end{proposition}
\begin{proof}
Consider the projective rank-one measurements $\{P^i_x\}$ obtained by performing $U_i$ followed by a measurement in the computational basis of $A$. Using the fact that the min-entropy cannot increase by too much when performing a measurement (Lemma~\ref{lem:qcvscc}), we obtain for all $i \in \{1, \dots, L\}$
\begin{align}
H_{\min}^{\sqrt{\eps}}(A|E)_{\rho} + \log |A| \geq \entHmin^{\sqrt{\eps}}(X_i|E)_{\rho}\ ,
\end{align}
where $X_i$ denotes the outcome of the measurement $\{P^i_x\}$. But condition~\eqref{eq:opt-ext-cond} implies that there exists $i \in \{1, \dots, L\}$ such that 
\begin{align}
\left\| \cT_{A\rightarrow A_{1}}\left( U_i\rho_{AE} U_i^{\dagger} \right) - \frac{\id_{A_1}}{|A_1|} \otimes \rho_E\right\|_1 \leq \eps.
\end{align}
By monotonicity of the min-entropy for classical registers~\cite[Lemma C.5]{Berta11}, we have that
\begin{align}
\entHmin^{\sqrt{\eps}}(X_i|E)_{\rho} \geq \entHmin^{\sqrt{\eps}}(A_1|E)_{\cT_{A\rightarrow A_{1}}\left( U_i\rho_{AE} U_i^{\dagger} \right)} \geq \log |A_1|\ ,
\end{align}
which proves the desired result.
\end{proof}

The following simple argument shows that the number of unitaries of a QC-extractor has to be at least about $1/\eps$.

\begin{proposition}[Lower bound on seed size]\label{prop:coneps}
Let $A=A_{1}A_{2}$. Any $(k,\eps)$-QC-extractor with $k \leq \log |A| - 1$ is composed of a set of unitaries on $A$ of size at least $L\geq1/\eps$.
\end{proposition}

\begin{proof}
Let $S \subseteq [|A_1|]$ be an arbitrary subset of $|A_1|/2$ basis elements of $A_1$. Then consider the state
\begin{align}
\rho_{A} = \frac{2}{|A|}\cdot\sum_{a_1 \in S, a_2 \in [|A_2|]} U_1^{\dagger} \proj{a_1a_2} U_1\ .
\end{align}
Note that $\cT(U_1 \rho_A U_1^{\dagger}) = \frac{2}{A_1}\sum_{a_1 \in S} \proj{a_1}$ and thus $\| \cT(U_1 \rho_A U_1^{\dagger}) - \frac{\id_{A_1}}{|A_1|} \|_1 = 1$. This implies the claim.
\end{proof}

Observe that in the case where the system $E$ is trivial (or classical), it was shown in~\cite{fawzi11} that there exists QC-extractors composed of $L = O(\log(1/\eps)\eps^{-2})$ unitaries. This is a difference with classical extractors for which the number of possible values of the seed has to be at least $\Omega((n - k) \eps^{-2})$ \cite{RTS00}. 


\subsection{Full set of mutually unbiased bases (MUBs)}\label{sec:QCFromAllMUBS}

We saw that unitary 2-designs define QC-extractors. As unitary 2-designs also define QQ-extractors, it is natural to expect that we can build smaller and simpler sets of unitaries if we are only interested in extracting random classical bits. In fact, in this section, we construct simpler sets of unitaries that define a QC-extractor. Two ingredients are used: a full set of mutually unbiased bases and a family of pair-wise independent permutations.\footnote{The idea of using permutations with mutually unbiased bases goes back to~\cite{Ind07} and was employed in~\cite{fawzi11} using results from~\cite{GUV09}. Permutation extractors were used in a classical context in~\cite{Reingold00} and state randomization with permutation extractors is discussed in~\cite{fawzi11}. The decoupling behaviour of (almost) pairwise independent families of permutations is discussed in~\cite{Szehr11}.}

A set of unitaries $\{U_1, \dots, U_L\}$ acting on $A$ is said to define \emph{mutually unbiased bases} if for all elements $\ket{a}, \ket{a'}$ of the computational basis of $A$, we have $| \bra{a'} U_j U_i^{\dagger} \ket{a} |^2 \leq |A|^{-1}$ for all $i \neq j$. In other words, a state described by a vector $U_i^{\dagger} \ket{a}$ of the basis $i$ gives a uniformly distributed outcome when measured in basis $j$ for $i \neq j$. For example the two bases, sometimes called computational and Hadamard bases (used in most quantum cryptographic protocols), are mutually unbiased. There can be at most $|A|+1$ mutually unbiased bases for $A$. Constructions of full sets of $|A|+1$ MUBs are known in prime power dimensions~\cite{WF89,boykin:mub}. Such unitaries can be implemented by quantum circuits of almost linear size; see \cite[Lemma 2.11]{fawzi11}. Mutually unbiased bases also have applications in quantum state determination~\cite{Ivo81, WF89}. 

To state our result, we will need one more notion. A family $\cP$ of permutations of a set $X$ is \emph{pair-wise independent} if for all $x_1 \neq x_2$ and $y_1 \neq y_2$, and if $\pi$ is uniformly distributed over $\cP$, $\prob{\pi(x_1) = y_1, \pi(x_2) = y_2} = \frac{1}{|X|(|X|-1)}$. If $X$ has a field structure, i.e., if $|X|$ is a prime power, it is simple to see that the family $\cP = \{ x \mapsto a \cdot x + b : a \in X^*, b \in X\}$ is pair-wise independent. In the following, permutations of basis elements of a Hilbert space $A$ should be seen as a unitary transformation on $A$.

\begin{theorem}\label{thm:full-set-mub}
Let $A=A_{1}A_{2}$ with $n = \log |A|$, $|A|$ a prime power, and consider the map $\cT_{A\rightarrow A_{1}}$ as defined in Equation~\eqref{eq:meas-map}. Then, if $\left\{U_1, \dots, U_{|A|+1}\right\}$ defines a full set of mutually unbiased bases, we have for $\delta\geq0$,
\begin{align}\label{eq:full-set-mub}
 \frac{1}{|\cP|}  \frac{1}{|A|+1} \sum_{P \in \cP} \sum_{i=1}^{|A|+1} \left\| \cT_{A\rightarrow A_{1}}\left(PU_i \rho_{AE} \left(PU_i\right)^{\dagger}\right) - \frac{\id_{A_{1}}}{|A_{1}|} \ox \rho_E \right\|_1
\leq \sqrt{\frac{|A_{1}|}{|A|+1} 2^{-H^{\delta}_{\min}(A|E)_{\rho}}}+2\delta\ ,
\end{align}
where $\cP$ is a set of pair-wise independent permutation matrices. In particular, the set $\{PU_i : P \in \cP, i \in [|A|+1]\}$ defines a $(k, \eps)$-QC-extractor provided
\begin{align}
\log |A_1| \leq n + k - 2 \log(1/\eps)\ ,
\end{align}
and the number of unitaries is
\begin{align}
L = (|A|+1) |\cP| = (|A|+1)|A| (|A|-1)\ .
\end{align}
\end{theorem}

The proof can be found in Appendix~\ref{app:proofs} and uses one-shot decoupling techniques~\cite{Berta08,frederic:decoupling,Wullschleger08,Szehr11,szehr:designs} together with ideas related to permutation extractors~\cite{fawzi11,Ind07,Szehr11}. Related theorems with the average taken only over a set of pairwise independent permutations were derived in~\cite{Szehr11}. The idea is to bound the trace norm in Equation~\eqref{eq:full-set-mub} by the Hilbert-Schmidt norm of some well-chosen operator. This term is then computed exactly using the fact that the set of all the MUB vectors form a \emph{complex projective} 2-design (Lemma~\ref{lem:mub-design}), and the fact that the set of permutations is pair-wise independent.


\subsection{Bitwise QC-extractor}\label{sec:singlequdit}

The unitaries we construct in this section are even simpler. They are composed of unitaries $V$ acting on single qudits followed by permutations $P$ of the computational basis elements. Note that this means that the measurements defined by these unitaries can be implemented with current technology. As the measurement $\cT$ commutes with the permutations $P$, we can first apply $V$, then measure in the computational basis and finally apply the permutation to the (classical) outcome of the measurement. In addition to the computational efficiency, the fact that the unitaries act on single qudits, is often a desirable property for the design of cryptographic protocols. In particular, the application to the noisy storage model that we present in Section \ref{sec:noisy} does make use of this fact.

Let $d \geq 2$ be a prime power so that there exists a complete set of mutually unbiased bases in dimension $d$. We represent such a set of bases by a set of unitary transformations $\left\{V_0,V_1,\dots,V_d\right\}$ mapping these bases to the standard basis. For example, for the qubit space ($d=2$), we can choose 
\begin{equation}
V_0 = \left( \begin{array}{cc}
1 & 0 \\
0 & 1
\end{array} \right)
\qquad V_1 = \frac{1}{\sqrt{2}} \left( \begin{array}{cc}
1 & 1 \\
1 & -1
\end{array} \right)
\qquad V_2 = \frac{1}{\sqrt{2}} \left( \begin{array}{cc}
1 & i \\
i & -1
\end{array} \right)\ .
\end{equation}
We define the set $\cV_{d,n}$ of unitary transformations on $n$ qudits by $\cV_{d,n}:=\left\{V=V_{u_1}\ox\cdots\ox V_{u_n}|u_i\in\left\{0,\dots,d\right\}\right\}$. As in the previous section, $\cP$ denotes a family of pair-wise independent functions.

\begin{theorem}\label{thm:singleQuditExtract}
Let $A=A_{1}A_{2}$ with $|A|=d^{n}$, $|A_{1}|=d^{\xi n}$, $|A_{2}|=d^{(1-\xi)n}$, and $d$ a prime power. Consider the map $\cT_{A\rightarrow A_{1}}$ as defined in Equation~\eqref{eq:meas-map}. Then for $\delta\geq0$ and $\delta'>0$,
\begin{align}
\notag
\frac{1}{|\cP|} \frac{1}{(d+1)^n} \sum_{P \in \cP}\sum_{V \in \cV_{d,n}} &\left\| \cT_{A\rightarrow A_{1}}\left(PV \rho_{AE} \left(PV\right)^{\dagger}\right)- \frac{\id_{A_{1}}}{|A_{1}|} \ox \rho_E \right\|_1 \\
&\leq \sqrt{2^{\left(1-\log (d+1)+\xi\log d\right)n}(1+2^{-H^{\delta}_{\min}(A|E)_{\rho}+z})}+2(\delta+\delta')\ ,
\label{eq:singleQuditExtract}
\end{align}
where $\cV_{d,n}$ is defined as above, $\cP$ is a set of pair-wise independent permutation matrices, and $z=\log\left(\frac{2}{\delta'^2}+\frac{1}{1-\delta}\right)$. In particular, the set $\{ PV : P \in \cP, V \in \cV_{d,n}\}$ is a $(k, \eps)$-extractor provided
\begin{align}
\log |A_1| \leq (\log(d+1) - 1) n + \min \left\{0, k\right\}  - 4 \log(1/\eps) - 7\, 
\end{align}
and the number of unitaries is
\begin{align}
L = (d+1)^n d^n (d^n - 1)\ .
\end{align}
\end{theorem}

The proof can be found in Appendix~\ref{app:proofs}. The analysis uses the same technique as in the proof of Theorem~\ref{thm:full-set-mub}. The main difference is that we were not able to express the Hilbert-Schmidt norm exactly in terms of the conditional min-entropy $\entHmin(A|E)_{\rho}$. Instead, we use some additional inequalities, which account for the slightly more complicated expression we obtain.

All results about QC-extractors are summarized in Table~\ref{tab:ext-summary} in the discussion section.


\section{Applications to Entropic Uncertainty Relations with Quantum Side Information}\label{sec:URbounds}

The first application of our result is to entropic uncertainty relations \emph{with quantum side information}. Entropic uncertainty relations form a modern way to capture the notion of uncertainty in quantum physics, and have interesting applications in quantum cryptography (see~\cite{ww:URsurvey} for a survey). Intuitively, uncertainty relations aim to address the following question. Let $\rho_K^j$ denote the distribution over classical outcomes $K$ given by measurement $j$ applied to some particular state $\rho_A$. Consider now a set of $L$ measurements that we could perform on some quantum system $A$. What is the allowed set of $L$ distributions $\rho_{K}^{j}$ for any quantum state $\rho_A$? Entropic uncertainty relations capture limitations to this allowed set by bounding the entropies of such distributions. Typically, they are stated as an average of entropies of the outcome distributions of the different measurements.

However, with regards to applications in quantum cryptography, it is important to realize that uncertainty should not be treated as absolute, but with respect to the prior knowledge of an observer $E$. This has far reaching consequences, as it comes to a subtle interplay between uncertainty and entanglement. The effect can be quantified by uncertainty relations \emph{with quantum side information}. Motivated by the case of two measurements~\cite{Berta09,Joe09,Coles11,Coles12,Coles10,christandl05,Tomamichel11} (i.e.~$L=2$), such relations should tell us that for all states $\rho_{AE}$, we have
\begin{align}\label{eq:entropicURGeneralForm}
	\frac{1}{L}\sum_{j=1}^L \hat{H}(K|E)_{\rho^{j}} \geq c + \tilde{H}(A|E)_{\rho}\ ,
\end{align}
where $\rho_{KE}^{j}=\meas^j_{A\rightarrow K}(\rho_{AE})$, $\hat{H}$ and $\tilde{H}$ are some conditional entropy measures, and $c$ is a constant depending on the choice of measurements $\{\meas^1_{A\rightarrow K},\ldots,\meas^L_{A\rightarrow K}\}$. Here the conditional entropy term on the rhs can in general become negative and is a measure for the entanglement in the pre-measurement state $\rho_{AE}$. Of particular interest are thereby the conditional von Neumann entropy, or (smoothed) R{\'e}nyi entropies.

In the case of classical side information ($E$ is a classical system), or no side information ($E$ is trivial), many such relations are known~\cite{ww:URsurvey}. Let us now first consider this case in 
more detail in order to recall some basic facts about entropic uncertainty relations. First of all, note that if $\hat{H}$ is the von Neumann entropy, one can use the chain rule to express the l.h.s. of Equation~\eqref{eq:entropicURGeneralForm}
as
\begin{align}
	H(K|J)_{\rho}=\frac{1}{L}\sum_{j=1}^L H(K)_{\rho^{j}}\ ,
\end{align}
where $\rho_{KJ}=\frac{1}{L}\sum_{j=1}^L \rho_{K}^{j} \otimes \proj{j}$, with $\rho_{K}^{j}=\meas^j_{A\rightarrow K}(\rho_{A})$ being the classical distribution over measurement outcomes when measurement $\meas^j_{A\rightarrow K}$ was performed on $\rho_A$. For other entropies we cannot simply rewrite the l.h.s.~in this manner, since no corresponding chain rule exists. Nevertheless, 
for most other interesting entropies, such as e.g.~the min-entropy, one can use the concavity of the $\log$ to lower bound
\begin{align}
	\frac{1}{L}\sum_{j=1}^L \hmin(K)_{\rho^{j}} &\geq
	- \log\left[\frac{1}{L} \sum_{j=1}^L 2^{- \hmin(K)_{\rho^{j}}}\right] = \hmin(K|J)_{\rho}\ .
\end{align}
In fact, this and most other entropies uncertainty relations are typically proven by lower bounding $\hat{H}(K|J)_{\rho}$ instead of the average. That is, existing proofs actually give us
\begin{align}
	\hat{H}(K|J)_{\rho} \geq c \ .
\end{align}
We will refer to such a relation as a \emph{meta-entropic} uncertainty relation. Meta-entropic relations are also the ones relevant for most quantum cryptographic applications and have a foundational significance in quantum information~\cite{js:urvsnl}.

Let us now return to the case \emph{with} quantum side information. The goal of this section is to show that QC-extractors lead to meta-entropic uncertainty relations with quantum side information of the form
\begin{align}
\hat{H}(K|EJ)_{\rho} \geq c + \tilde{H}(A|E)_{\rho}\ .
\end{align}


\subsection{Idea}

Our approach of using QC-extractors to derive strong entropic uncertainty relations is based on ideas developed in~\cite{fawzi11}. In fact, as outlined in Section~\ref{sec:QCFromGeneralDecoupling}, one can understand the set of unitaries constructed in~\cite{fawzi11} as QC-extractors \emph{without} quantum side information. As opposed to~\cite{fawzi11}, we start with uncertainty relations for the smooth conditional min-entropy, since this is the relevant operational quantity to bound for quantum cryptographic applications. We first prove a meta-uncertainty relation which is essentially immediate.

\begin{lemma}\label{lem:minentropy}
Let $\rho_{AE}\in\states(AE)$, and $\left\{U_1, \dots, U_{L}\right\}$ be a set of unitaries on $A$ with corresponding measurements $\left\{\cM_{A\rightarrow K_{1}}^{1},\ldots,\cM_{A\rightarrow K_{1}}^{L}\right\}$ as defined in Equation~\eqref{eq:meas}, such that,
\begin{align}	
\frac{1}{L}\sum_{j=1}^{L}\left\|\cM_{A\rightarrow K_{1}}^{j}(\rho_{AE})-\frac{\id_{K_{1}}}{|K_{1}|}\otimes\rho_{E}\right\|_{1}\leq\eps(\rho)\ ,
\end{align}
for some $\eps(\rho)$ depending on the input state $\rho_{AE}$. Then
\begin{align}
\hmin^{\sqrt{2\eps(\rho)}}(K_{1}|EJ)_{\rho} \geq \log|K_{1}|\ ,
\end{align}
where $\rho_{K_{1}EJ}=\frac{1}{L}\sum_{j=1}^{L}\meas^j_{A\rightarrow K_{1}}(\rho_{AE})\otimes\proj{j}_{J}$.
\end{lemma}

\begin{proof}
	Note that since $\|\rho_{K_{1}EJ}-\iota_{K_{1}EJ}\|_{1}\leq\eps(\rho)$ with $\iota_{K_{1}EJ}=\frac{\id_{K_{1}}}{|K_{1}|}\otimes\rho_{EJ}$, we have by Equation~\eqref{eq:purifiedVStrace} applied to normalized states that
	$P(\rho_{K_{1}EJ},\iota) \leq \sqrt{2 \eps(\rho)}$.
	It then follows immediately from the definition of the smooth conditional min-entropy (Equation~\eqref{eq:smoothMinDef}) that 
	$\hmin^{\sqrt{2\eps(\rho)}}(K_{1}|EJ)_{\rho} \geq \hmin(K_{1}|EJ)_{\iota}$.
	Our claim now follows by noting that $\hmin(K_{1}|EJ)_{\iota} = \hmin(K_{1})_{\iota} = \log |K_{1}|$.
\end{proof}



Uncertainty relations for the conditional von Neumann entropy can be obtained as follows.

\begin{lemma}\label{lem:shannon}
For the same premises as in Lemma~\ref{lem:minentropy}, we have
	\begin{align}
		\frac{1}{L} \sum_{j=1}^{L} H(K_{1}|E)_{\rho^{j}}=H(K_{1}|EJ)_{\rho} \geq (1-4\eps(\rho))\log|K_{1}|-2h(\eps(\rho))\ ,
	\end{align}
	where $\rho_{K_{1}E}^{j}=\meas^j_{A\rightarrow K_{1}}(\rho_{AE})$ and $\rho_{K_{1}EJ}=\frac{1}{L}\sum_{j=1}^{L}\rho_{K_{1}E}^{j}\otimes\proj{j}_{J}$.
\end{lemma}

\begin{proof}
By assumption we have
\begin{align}
\left\|\rho_{K_{1}EJ}-\frac{\id_{K_{1}}}{|K_{1}|}\otimes\rho_{EJ}\right\|_{1}\leq\eps(\rho)\ ,
\end{align}
and hence the improved Alicki-Fannes inequality~\cite{Alick04} immediately implies that the claim.
\end{proof}

So far, we have merely made a few rather simple statements, and it is not easy to see how these uncertainty relations should at all take quantum side information into account. This link is forged by the exact form of the approximation parameter $\eps(\rho)$ for QC-extractors. We again start with the min-entropy case.


\subsection{Uncertainty relations for the min-entropy}\label{sec:mini}

To get some intuition of how our line of proof works, let us now consider two simple examples in detail - bounds for all constructions are summarized in Table~\ref{tab:URsummary}.


\subsubsection{Exact unitary $2$-designs}

As an illustrative warmup, we consider measurements formed by applying a unitary $U_{j}$ drawn from an exact unitary $2$-design to $A=A_{1}A_{2}$, followed by a measurement of $A$ in the standard basis leading to a classical outcome register $K$. Denote this measurement by $\cM^{j}_{A\rightarrow K}$. Note that we can perform the measurement in two steps. First, we measure $A_{1}$ to obtain a classical outcome $K_{1}$. Second, we measure $A_{2}$ to obtain a classical outcome $K_{2}$.  Let us first consider only the outcome $K_{1}$, tracing over the resulting classical register $K_{2}$. This then corresponds to the measurements $\cM_{A\rightarrow K_{1}}^{j}$ generated by the unitaries $U_{j}$ (cf.~Equation~\eqref{eq:meas}) drawn from the exact unitary $2$-design.

As outlined in Section~\ref{sec:QCFromGeneralDecoupling} and Lemma~\ref{lem:2design}, the general decoupling results of~\cite{frederic:decoupling,Wullschleger08} immediately imply that the set of such measurements forms a QC-extractor with
\begin{align}
	\eps(\rho)=2^{-\frac{1}{2}\left(\hmin^{\delta}(A|E)_\rho + \log|A_{2}|\right)}+2\delta\ ,
\end{align}
for any $\delta\geq0$. Intuitively, $\eps(\rho)$ becomes larger if $E$ is highly entangled with $A$ and smaller if we trace out a larger chunk $A_{2}$ from the initial system. Let us suppose we would like to have an entropic uncertainty relation with respect to quantum side information for some particular fixed $\eps'= \eps(\rho)$. Do there exist measurements that give us such a high amount of uncertainty? Our results from the previous section tell us that such measurements do indeed exist, if we choose $A_{2}$ from the combined system $A=A_{1}A_{2}$ large enough. In particular, choose $A_{2}$ such that 
\begin{align}\label{eq:decoup}
	|A_{2}| = \frac{1}{\left(\eps'-2\delta\right)^2}\cdot2^{-\hmin^{\delta}(A|E)_\rho}\ .
\end{align}
Using that $\log|K_{1}| = \log |A_{1}| = \log|A| - \log|A_{2}|$ we have by Lemma~\ref{lem:minentropy} and the monotonicity of the min-entropy for the classical register $K_{2}$~\cite[Lemma C.5]{Berta11} that,
\begin{align}
	\hmin^{\sqrt{2\eps'}}(K|EJ)_{\rho}\geq\hmin^{\sqrt{2\eps'}}(K_{1}|EJ)_{\rho}\geq \log|A|-\log\left(\frac{1}{\left(\eps'-2\delta\right)^{2}}\right)+\hmin^{\delta}(A|E)_{\rho}\ ,
\end{align}
where $\rho_{KEJ}=\frac{1}{L}\sum_{j}\cM^{j}_{A\rightarrow K}(\rho_{AE})\otimes\proj{j}_{J}$. We set $\eps'=\eps^{2}/2$ and conclude that for any $\eps>0$ and $\delta\geq0$,
\begin{align}\label{eq:ws}
\hmin^{\eps}(K|EJ)_{\rho}\geq\log|A|-\log\left(\frac{1}{\left(\eps^{2}/2-2\delta\right)^{2}}\right)+\hmin^{\delta}(A|E)_{\rho}\ .
\end{align}
Note that since l.h.s. is in fact upper bounded by $\log|A| = \log|K|$, this uncertainty relation is very strong as long as $|A|$ is sufficiently large. To gain some intuition about this bound, consider the case of trivial side information $E$ for which $\hmin(A|E)_{\rho}=\hmin(A)_{\rho}\geq0$.


\subsubsection{Full set of MUBs}

As the second example we consider a measurement of $A=A_{1}A_{2}$ in the full set of $|A|+1$ MUBs. As mentioned before, it is known that whenever $|A|= p^k$ with $p$ prime, such a set exists~\cite{WF89,boykin:mub}. As before, we denote the classical outcome of measuring $A=A_{1}A_{2}$ with $K=K_{1}K_{2}$. Let us now first consider a post-processing of this measurement. In particular, suppose that we randomly choose a two-wise independent permutation $\pi$ over $|A|$ to obtain the string $\Pi(K_{1}K_{2})$, and trace out a system of size $\log|A_{2}|$ at the end. Let $K_{\Pi}$ denote the resulting string. Note that this can be understood as a new set of measurements. As shown in Theorem~\ref{thm:full-set-mub} these form a QC-extractor with
\begin{align}\label{eq:fullSetURExample}
	\eps(\rho)=\sqrt{\frac{|A_{1}|}{|A|+1} 2^{\hmin^{\delta}(A|E)_{\rho}}}+2\delta\ ,
\end{align}
for any $\delta\geq0$. To obtain a meaningful uncertainty relation, let us now again fix some particular $\eps'=\eps(\rho)$ and choose $|A_{1}|$ accordingly. By rearranging the terms in~\eqref{eq:fullSetURExample} we obtain
\begin{align}\label{eq:KvalueMUBS}
	\log |A_{1}| = 
	\log\left(|A|+1\right)-\log\left(\frac{1}{(\eps'-2\delta)^2}\right)+\hmin^{\delta}(A|E)_{\rho}\ .
\end{align}
As above, one may now immediately write down uncertainty relations for the new, larger, set of measurements. Intuitively, it is clear however that adding the additional permutation does not change matters - after all entropic measures only depend on the distributions and are invariant under relabelings of the actual symbols. This can be seen more formally as well by noting that $\Pi(K_{1}K_{2})$ can be computed from $K_{1}K_{2}$ and $\Pi$. We thus have for $\eps=\sqrt{2\eps'}>0$ and $\delta\geq0$ that,
\begin{align}
\hmin^{\eps}(K|EJ)_{\rho}&\geq \hmin^{\eps}(K|EJ\Pi)_{\rho}=\hmin^{\eps}(\Pi(K)|EJ\Pi)_{\rho}\geq \hmin^{\eps}(K_{\Pi}|EJ\Pi)_{\rho}\\
&\geq \log\left(|A|+1\right)-\log\left(\frac{1}{\left(\eps^{2}/2-2\delta\right)^{2}}\right)+\hmin^{\delta}(A|E)_{\rho}\ ,
\end{align}
where the first step follows because conditioning reduces the min-entropy~\cite[Theorem 18]{Tomamichel09}, the second from the fact that $\Pi(K)$ can be computed from $K$ and $\Pi$~\cite[Lemma 13]{Tomamichel09}, the third from the monotonicity of the min-entropy for a classical register~\cite[Lemma C.5]{Berta11}, and the last one from Lemma~\ref{lem:minentropy} and Equation~\eqref{eq:KvalueMUBS}.


\subsubsection{Single-qudit measurements}

From the point of view of applications, the following entropic uncertainty relation for single-qudit measurements is probably the most interesting. It can be seen as a generalization to allow for quantum side information of uncertainty relations obtained in \cite{serge:new}.
\begin{theorem}
\label{thm:min-ent-ur-bitwise}
Let $d \geq 2$ be a prime power. For any state $\rho_{AE}$, we have
\[
\hmin^{\eps}(K|EJ)_{\rho} \geq n\cdot\left(\log(d+1)-1\right) + \min\left\{0,\hmin^{\delta}(A|E)_{\rho}-\log\left(\frac{2}{\delta'^2}+\frac{1}{1-2\delta}\right)\right\}-\log\left(\frac{1}{\left(\eps^{2}/2-2\delta-\delta'\right)^2}\right)-1,
\]
where $\rho_{KEJ} = \frac{1}{(d+1)^n} \sum_{j} \cM_{A \to K}^j \otimes \proj{j}_J$ and the measurements $\cM_{A \to K}^j$ correspond to measuring each qudit in a basis from a set of MUBs.
\end{theorem}

We summarize the various uncertainty relations relations for the min-entropy in the following table.

\begin{table}[ht]
\label{tbl:min-ent-ur}
\centering
\begin{tabular}{|l|c|}
	\hline
	& Lower bounds for the smooth conditional min-entropy $\hmin^{\eps}(K|EJ)_{\rho}$\\
	\hline
	Unitary 2-design & $\log|A|+\hmin^{\delta}(A|E)_{\rho}-\log\left(\frac{1}{\left(\eps^{2}/2-2\delta\right)^{2}}\right)$\\
	\hline
	Almost unitary 2-design &  $\log|A|+\hmin^{\delta}(A|E)_{\rho}-\log\left(\frac{1}{\left(\eps^{2}/2-2\delta\right)^{2}}\right)-\log\left(1+\zeta\right)$\\
	\hline
	All $|A|+1$ MUBs
	&
	$\log\left(|A| + 1\right)+\hmin^{\delta}(A|E)_{\rho}-\log\left(\frac{1}{\left(\eps^{2}/2-2\delta\right)^{2}}\right)$
	\\
	\hline 
	Single qudit MUBs &  
	$n\cdot\left(\log(d+1)-1\right) + \min\left\{0,\hmin^{\delta}(A|E)_{\rho}-\log\left(\frac{2}{\delta'^2}+\frac{1}{1-2\delta}\right)\right\}-\log\left(\frac{1}{\left(\eps^{2}/2-2\delta-\delta'\right)^2}\right)-1$
	\\
	\hline
\end{tabular}
\caption{Entropic uncertainty relations with quantum side information for the smooth conditional min-entropy for approximation parameters $\eps>0$, $\zeta\geq0$, $\eta>0$, $\delta\geq0$, and $\delta'>0$. The almost unitary 2-design has an approximation parameter of $\frac{\zeta}{4|A|^{4}}$ and can be sampled from using a random quantum circuit of size $O\left(\log|A|\left(\log|A|+\log\left(\frac{1}{\zeta}\right)\right)\right)$~\cite{Harrow09_2,szehr:designs,Szehr11}.}
\label{tab:URsummary}
\end{table}


\subsection{Uncertainty relations for the von Neumann entropy}

Let us now turn to entropic uncertainty relations in terms of the von Neumann entropy. To this end, we again consider the same two examples. 

\subsubsection{Exact unitary $2$-designs}

First, we again consider a set of measurements given by a unitary $2$-design. Using~\eqref{eq:decoup} and the fact that $\log|K_{1}| = \log |A_{1}| = \log|A| - \log|A_{2}|$, we obtain from by the monotonicity of the von Neumann entropy and Lemma~\ref{lem:shannon} that for $\eps>0$ and $\delta\geq0$,
\begin{align}
\frac{1}{L}\sum_{j=1}^{L} H(K|E)_{\rho^{j}}\geq \frac{1}{L}\sum_{j=1}^{L} H(K_{1}|E)_{\rho^{j}}\geq(1-4\eps)\left(\log|A|+\hmin^{\delta}(A|E)_{\rho}-\log\left(\frac{1}{\left(\eps-2\delta\right)^{2}}\right)\right) - 2h(\eps)\ .
\end{align}


\subsubsection{Full set of MUBs and single qudit MUBs}

Similarly for the full set of $|A|+1$ MUBs, we get as in the min-entropy case that for $\eps>0$ and $\delta\geq0$,
\begin{align}
	\frac{1}{L} \sum_j H(K|E)_{\rho^{j}}&=H(K|EJ)_{\rho}\geq H(K|EJ\Pi)=H(\Pi(K)|EJ\Pi)\geq H(K_{\Pi}|EJ\Pi)\\
	&\geq (1-4\eps)\left(\log\left(|A| + 1\right) + \hmin^{\delta}(A|E)_{\rho} - \log\left(\frac{1}{(\eps-2\delta)^2}\right)\right) - 2h(\eps)\ ,
\end{align}
where the first step follows from the chain rule for the von Neumann entropy, the second from the fact that conditioning reduces entropy, the third because $\Pi(K)$ can be computed from $K$ and $\Pi$, the fourth from the monotonicity of the von Neumann entropy for a classical register, and the last from Lemma~\ref{lem:shannon} together with Equation~\eqref{eq:KvalueMUBS}.

Glancing at both uncertainty relations, it seems rather unsatisfying that we have a mix of entropies. That is, on the left we quantify information in terms of the von Neumann entropy, whereas on the right we employ the min-entropy. Can we derive a relation solely in terms of the von Neumann entropy? As we prove in Appendix~\ref{sec:urProofs} this is indeed the case, where we use the fact that the smooth min-entropy approaches the von Neumann entropy in the asymptotic limit of many copies of the state (Lemma~\ref{lem:aep}).

\begin{proposition}\label{thm:neumann}
Let $d \geq 2$ be a prime power, and $\left\{V_0,V_1,\dots,V_d\right\}$ define a complete set of MUBs of $\CC^d$. Consider the set of measurements $\{\cM^j_{A \to K} : j \in [(d+1)^n]\}$ on the $n$ qudit space $A$ defined by the unitary transformations $\left\{V=V_{u_1}\ox\cdots\ox V_{u_n}|u_i\in\left\{0,\dots,d\right\}\right\}$. Then for all $\rho_{AE} \in \cS(AE)$, we have
\begin{align}
\frac{1}{(d+1)^n}\sum_{j=1}^{(d+1)^n}H(K|E)_{\rho^{j}}\geq n\cdot\left(\log(d+1)-1\right)+\min\left\{0,H(A|E)_{\rho}\right\}\ ,
\end{align}
where $\rho^j = \cM^j_{A \to K}(\rho)$.
\end{proposition}

Note that for $n=1$, this again gives an uncertainty relation for the full set of MUBs, but now a \lq complete\rq~von Neumann entropy version
\begin{align}\label{eq:finalvonneumman2}
\frac{1}{d+1}\sum_{j=1}^{d+1}H(K|E)_{\rho^{j}}\geq \log(d+1)-1+\min\left\{0,H(A|E)_{\rho}\right\}\ .
\end{align}
To understand this bound it is again instructive to consider some special cases.  Note that for $E$ trivial, we arrive at
\begin{align}\label{eq:sanchez_H2}
\frac{1}{d+1}\sum_{j=1}^{d+1}H(K)_{\rho^{j}}\geq \log(d+1)-1\ ,
\end{align}
which is the best known bound for a full set of MUBs and general $d$~\cite{Ivanovic92,Sanchez93}. But it is also known that without side information and $d$ even, this can be improved to~\cite{Sanchez95}
\begin{align}
\frac{1}{d+1}\sum_{j=1}^{d+1}H(K)_{\rho^{j}}\geq \frac{1}{d+1} \left(\frac{d}{2}\log\left(\frac{d}{2}\right)+\left(\frac{d}{2}+1\right)\log\left(\frac{d}{2}+1\right)\right)\ .
\end{align}
For one qubit ($d=2$) the latter gives $2/3$ (which is known to be tight, e.g. for the Pauli matrices), whereas our bound gives $\log3-1\approx0.585$. 


\subsection{Conclusions}

Previously, uncertainty relations \emph{with} quantum side information were only known for two measurements~\cite{Berta09,Joe09,Coles11,Coles12,Coles10,christandl05,Tomamichel11}. As shown above, however, \emph{any} QC-extractor yields an uncertainty relation that takes quantum side information into account. Tables~\ref{tab:URsummary} summarizes the uncertainty relations for the min-entropy obtained for the particular QC-extractors from this paper. For the von Neumann entropy uncertainty relations, we would mainly like to point to Proposition~\ref{thm:neumann} and Equation~\eqref{eq:finalvonneumman2}, which can be understood as the generalization of a well known entropic uncertainty relation \emph{without} quantum side information (Equation~\eqref{eq:sanchez_H2}).


\section{Applications to Security in the Noisy-Storage Model}\label{sec:noisy}

As the second application, we solve the long standing question of relating the security of cryptographic protocols in the noisy-storage model~\cite{Noisy1,noisy:robust,Curty10,noisy:new} to the quantum capacity.

\subsection{Model}

Let us first provide a brief summary of the noisy-storage model - details can be found in~\cite{noisy:new}. The central assumption of the model is that during waiting times $\Delta t$ introduced into the protocol, the adversary can only store quantum information using a limited and unreliable quantum memory device. This is indeed the only assumption on the adversary who is otherwise all powerful. In particular, he can store an unlimited amount of classical information, and perform any operation instantaneously. The latter implies that he is able to perform any encoding and decoding operations before and after using his memory device, even if these may be difficult to perform.

Mathematically, such a quantum storage device is simply a quantum channel $\cF:\bop(\hil_{\rm in}) \rightarrow \bop(\hil_{\rm out})$ mapping input states on the space $\hil_{\rm in}$ to some noisy output states on the space $\hil_{\rm out}$. Of particular interest are thereby input spaces of the form $\hil_{\rm in} = (\Complex^d)^{\otimes N}$ and channels $\cF = \cN^{\otimes N}$ with $\cN: \bop(\hil_{\rm in}) \rightarrow \bop(\hil_{\rm out})$. This corresponds to a memory device consisting of $N$ $d$-dimensional \lq memory cells\rq~ each of which experiences a noise described by the channel $\cN$. A special case of this model is thus the bounded quantum storage model where $d = 2$, 
and $\cF=\cI_2$ is the one qubit identity channel~\cite{serge:new,serge:bounded}. For a protocol using BB84 encoded qubits it is known that security can be achieved whenever $N$ is strictly less than half the number of qubits sent during the course of the protocol~\cite{noisy:new}.


\subsection{Security of existing protocols}

\subsubsection{Weak string erasure}

How can we hope to show security in such a model? In~\cite{noisy:new} it was shown that bit commitment and oblivious transfer, and hence any two-party secure 
computation~\cite{kilian}, can be implemented securely against an \emph{all-powerful} quantum adversary given access to a much simpler primitive
called \emph{weak string erasure (WSE)}. The latter primitive was then proven secure in the noisy-storage model. It is hence enough to prove the security of WSE, and we will follow this approach here. 

The motivation behind the primitive weak string erasure was to create a basic \emph{quantum} protocol that builds up \emph{classical} correlations between Alice and Bob which are 
later used to implement more interesting cryptographic primitives. Informally, weak string erasure achieves the following task - a formal definition~\cite{noisy:new,prabha:limits} can be 
found in the appendix. WSE takes no inputs from Alice and Bob. Alice receives as output a randomly chosen string $X^n = X_1,\ldots,X_n \in \{0,1\}^n$. Bob
receives a randomly chosen subset $\cI \in [n]$ and the substring $X_{\cI}$ of $X^n$. Randomly chosen thereby means that each index $i \in [n]$ has some fixed 
probability $p$ of being in $\cI$. Originally, $p=1/2$~\cite{noisy:new}, but any probability $0 < p <1$ allows for the implementation of oblivious 
transfer~\cite{prabha:limits}. The security requirements of weak string erasure are that Alice does not learn $\cI$, and Bob's min-entropy given all of his information $B$ is 
bounded as $\hmin(X|B) \geq \lambda n$ for some parameter $\lambda > 0$. To summarize all relevant parameters, we thereby speak of an $(n,\lambda,\eps,p)$-WSE scheme.


\subsubsection{Protocol for weak string erasure}

We now construct a very simple protocol for weak string erasure, and prove its security using our bitwise QC-randomness extractor. 
The only difference to the protocol proposed in~\cite{noisy:new} is that we will use 3 MUBs per qubit instead of only $2$.
For sake of argument, we state the protocol in a purified form where Alice generates EPR-pairs and later measures them. Note, however, that the protocol
is \emph{entirely} equivalent to Alice creating single qubits and sending them directly to Bob. In the purified protocol, the choice of bit she encodes is determined randomly 
by her measurement outcome in the chosen basis on the EPR-pair. 
That is, honest Alice and Bob do not need any quantum memory to implement the protocol below.
Indeed, this is the way such protocols are typically implemented in practice.

\begin{protocol}{Weak string erasure (WSE)}{Outputs: $x^n \in \01^n$
	to Alice, $(\cI,z^{|\cI|}) \in 2^{[n]} \times \01^{|\cI|}$ to Bob.}{}
\item[1.] {\bf Alice:} Creates $n$ EPR-pairs $\Phi$, and sends half of each pair to Bob.
\item[2.] {\bf Alice:} 
	Chooses a bases-specifying string $\theta^n \in_R \{0,1,2\}^n$ uniformly at random. 
	For all $i$, she measures the $i$-th qubit in the basis $\theta_i$ to obtain outcome $x_i$.

\item[3.] {\bf Bob: } Chooses a basis string $\tilde{\theta}^n \in_R \{0,1,2\}^n$ uniformly at random. When receiving the $i$-th qubit, Bob measures it in the basis
given by $\tilde{\theta}_i$ to obtain outcome $\tilde{x}_i$.

\item[Both parties wait time $\Delta t$.]

\item[4.] {\bf Alice: } Sends the basis information $\theta^n$ to Bob, and outputs $x^n$.
\item[5.] {\bf Bob: } Computes $\iSet = \{i \in [n] \mid \theta_i = \tilde{\theta}_i\}$, and outputs $(\cI,z^{|\cI|}):=(\cI,\tilde{x}_{\cI})$.
\end{protocol}

The proof of correctness of the protocol, and security against dishonest Alice is identical to~\cite{noisy:new,prabha:limits}. It essentially follows from the fact that
Bob never sends any information to Alice. The main difficulty lies in 
proving security against dishonest Bob. Before embarking on a formal proof, let us first consider the general form that \emph{any} attack of Bob takes (see Figure~\ref{fig:BobAttack}).
First of all, note that the noisy-storage model only assumes that Bob has to use his storage device during waiting times $\Delta t$. That is, when attacking the protocol above he can in fact store the incoming
qubits perfectly
until the waiting time, i.e., until all $n$ qubits arrived. Let $Q$ denote Bob's quantum register containing all $n$ qubits. 
Note that since there is no communication between Alice and Bob during the transmission of these $n$ qubits, we can without loss of generality 
assume that Bob first waits for all $n$ qubits to arrive before mounting any form of attack.

As any operation in quantum theory is a quantum channel, Bob's attack can also be described by a quantum channel $\cE: \substates(Q) \rightarrow \substates(\hin \otimes \msg)$. This map takes $Q$, to some quantum state on the input of Bob's storage device ($\hin$), and some arbitrarily large amount of classical information ($\msg$). For example, $\cE$ could be an encoding into an error-correcting code. By assumption of 
the noisy-storage model, Bob's quantum memory is then affected by noise $\cF: \substates(\hin) \rightarrow \substates(\hout)$.
After the waiting time, the joint state held by Alice and Bob in the purified version of the protocol, i.e., before Alice measures, is thus of the form
\begin{align}\label{eq:rho}
	\rho_{ABM} = \cI_A \otimes \left[\left(\cF \otimes \cI_{\msg}\right) \circ \cE\right](\Phi^{\otimes n})\ ,
\end{align}
where $\Phi$ is an EPR-pair.
After the waiting time, Bob can perform any form of quantum operation to try and recover information from the storage device. 
Note that in principle, Bob's goal is to recover $X$ alone for which he could potentially use his basis information $\Theta$. 
Yet, we will see in Section~\ref{sec:securityCap} that we can ignore the basis information in the analysis. That is, we only need
to analyze decoding maps $\cD: \substates(\hin \otimes \msg) \rightarrow \substates(Q')$ trying to recover the initial entanglement between Alice and Bob. 
\begin{center}
	\begin{figure}
		\includegraphics{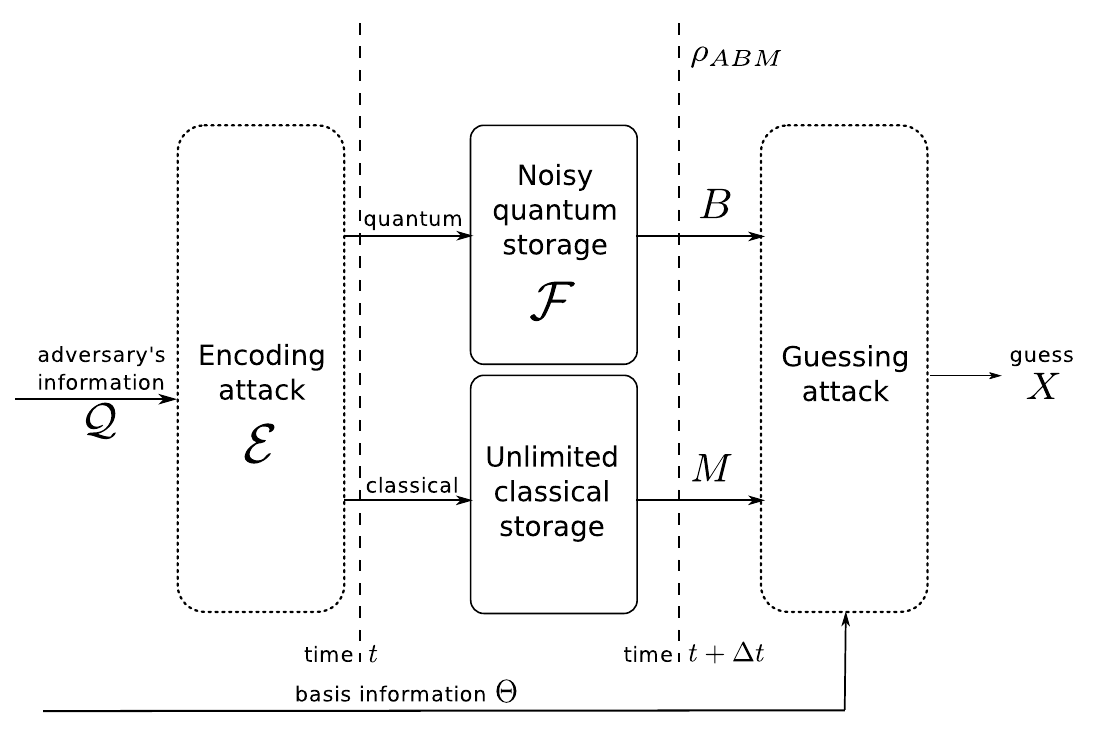}
		\caption{Any attack of dishonest Bob is described by an encoding attack $\cE$ and a \lq guessing\rq~attack, since for classical $X$ the min-entropy $\hmin(X|BM\Theta)$ is
		directly related to the probability that Bob guesses $X$. As we will see below, it is however sufficient to consider how well a decoding attack $\cD$ can preserve entanglement
		between Alice and Bob, where $\cD$ acts on $BM$ on the state $\rho_{ABM}$ from Equation~\eqref{eq:rho} at the marked point in time.}
		\label{fig:BobAttack}
	\end{figure}
\end{center}


\subsection{Security and the quantum capacity}\label{sec:securityCap}

Recall from the definition above that our goal is to show that $\hmin^\eps(X|BM\Theta)_{{\rho}} \geq \lambda \cdot n$ for some parameter $\lambda$. How could we hope to accomplish this?
Although it was always clear that security should be related to the channel's ability to store quantum information, i.e., the \emph{quantum capacity} of $\cF$, proving this fact has
long formed an elusive problem. Partial progress to answering this question was made in~\cite{noisy:new} and~\cite{entCost}, where security was linked to the \emph{classical capacity} and \emph{entanglement cost}
of $\cF$, respectively. Why would this problem be difficult? Note that we wish to make a statement about some \emph{classical} information $X$ obtained by measuring $A$ in bases $\Theta$.
That is, we effectively ask for an uncertainty relation for said measurements. Previously, however, suitable uncertainty relations were only known for \emph{classical} side information. 
The missing ingredient is thus an uncertainty relation with \emph{quantum} side information, linked to the channel's ability to preserve quantum information.

Indeed, one application of our QC-extractors is to provide such a relation, where for the protocol above we will need the relation for $3$ MUBs per qubit given in Table~\ref{tab:URsummary}. For $E=BM$ on ${\rho}_{ABM\Theta}$ it reads
\begin{align}\label{eq:minBound}
	\hmineps(X|BM\Theta)_{{\rho}} \gtrsim (\log(3)-1) n + \min\{0,\hmin(A|BM)_{{\rho}}\}\ .
\end{align}
Note that the operational definition of the smooth conditional min-entropy already incorporates any guessing attack Bob may mount on $BM\Theta$.
Clearly, not all QC-extractors are useful for protocols such as the above, as we must ensure that there exists a strategy for the honest players to succeed.
However, any \emph{bitwise} QC-extractor will do.
How can we now relate this expression to the quantum capacity? Note that the min-entropy has an appealing operational interpretation~\cite{krs:operational}
as
\begin{align}\label{eq:minOp}
	\hmin(A|BM)_{{\rho}} = - \log |A| \max_{\Lambda_{BM\rightarrow A'}} F(\Phi_{AA'},\id_A \otimes \Lambda(\rho_{ABM}))\ ,
\end{align}
where $\Phi_{AA'}$ is the maximally entangled state accross $AA'$. That is, the min-entropy is directly related to the \lq amount\rq~of entanglement between $A$ and $E=BM$. 
To place a bound on~\eqref{eq:minBound}, we would like to obtain a lower bound on
\begin{align}
	\min_{\cE} \hmin(A|BM)_{{\rho}}\ ,
\end{align}
where the minimization is taken over all encoding attacks described above. Note that this expression does not depend on the basis information $\Theta$, and
that the map $\Lambda$ in~\eqref{eq:minOp} can be understood as a decoding attack $\cD$ aiming to restore entanglement with Alice.
Further, note that $|A'| = |Q|$ and we can equivalently upper bound
\begin{align}\label{eq:chanFid}
	\max_{\cD,\cE} F\left(\Phi_{AB},\id_A \otimes \left[\cD \circ (\cF \otimes \cI_{\msg}) \circ \cE\right](\Phi_{AQ})\right) = 
	\max_{\cD,\cE} F_c(\cD \circ (\cF \otimes \cI_{\msg}) \circ \cE)\ .
\end{align}
The quantity $F_c$ on the r.h.s., however, is precisely the \emph{channel fidelity}~\cite{Barnum00} of $\cD \circ (\cF \otimes \cI_{\msg}) \circ \cE$, maximized over all encodings and decodings
where we are allowed free forward classical communication ($\msg$). 

Why is this quantity interesting? When talking about a channel's ability to carry information, we need to agree on what it means to send information reliably. 
The channel fidelity is one of the measures in which the quantum capacity can be expressed~\cite{variazoni}.
For the storage device $\cF$, the quantity 
\begin{align}\label{eq:defcap}
	&n = \max \log |A| \\
	&s.t. \max_{\cD,\cE} F_c(\cD \circ (\cF \otimes \cI_{\msg}) \circ \cE) \geq 1 - \eps\ ,
\end{align}
tells us how much entanglement, or equivalently how many qubits~\cite{Barnum00}, we can send through $\cF$ with an error of at most $\eps$, using free feed forward classical 
communication ($\msg$).  For $\eps \rightarrow 0$, this quantity is also known as the \emph{one-shot} quantum capacity $Q_{\rightarrow}^{(1)}$ of $\cF$ itself, no matter what form $\cF$ takes. 

Let us now consider storage devices of the form $\cF = \cN^{\otimes N}$.
Recall that the capacity of the channel $\cN$ is the maximum rate $R=n/N$ at which we can send $n$ (qu)bits reliably by using the channel $\cN$ $N$ times. 
For channels $\cF = \cN^{\otimes N}$, the quantity $R = n/N$ with $n$ from Equation~\eqref{eq:defcap} thus determines the maximum rate at which we can send information with error $\eps$ for any finite $N$. 
The usual quantum capacity with classical feed forward communication $Q_{\rightarrow}(\cN)$ is then
given by taking the limit $N\rightarrow \infty$ and $\eps \rightarrow 0$.

Whereas one might think that forward classical communication helps, it is in fact known that 
it does not affect the quantum capacity since for any scheme that achieves error $\eps$ using classical forward communication, there exists
a scheme without any classical communication with error at most $2\eps$~\cite{Barnum00}. Note that there are several definitions of the quantum capacity 
using e.g. the entanglement fidelity or the distance from the identity channel in diamond norm as a measure of success, however, all such definitions lead to the same capacity~\cite{variazoni}.
Combining Equation~\eqref{eq:minBound} and Equation~\eqref{eq:chanFid} thus finally relates the security in the noisy-storage model to the \emph{quantum} capacity $Q_{\rightarrow}(\cN)$ of the storage 
device.\footnote{Note that this also relates security to the one-shot capacity $Q^{(1)}_{\rightarrow}(\cF)$ of an arbitrary channel $\cF$.}


\subsection{Security parameters from a strong converse}\label{sec:securityStatement}

How can we now obtain explicit security parameters from this? We first make a statement analogous to~\cite[Theorem III.2.i]{noisy:new} for arbitrary channels $\cF$. 

\begin{theorem}\label{thm:para}
	Let Bob's storage device be given by $\cF$. 
	For any choice of constant parameters $\eps,\delta' > 0$, 
	Protocol 1 implements $(n,\lambda,\eps,1/3)$-WSE with 
	\begin{align}
		\lambda = \log(3) - 1 - \frac{1}{n}\max\left\{0,\max_{\cD,\cE}\log 2^n F_c(\cD \circ (\cF \otimes \cI_{\msg}) \circ \cE) + \kappa\right\}
		- \frac{1}{n}\left(\xi + 1 \right)\ ,
	\end{align}
	where $\kappa = \log\left(2/\delta'^2+1\right)$ and $\xi = \log\left(1/\left(\eps^{2}/2-\delta'\right)^2\right)$.
\end{theorem}

\begin{proof}
	The proof of correctness of the protocol, and security against dishonest Alice is identical to~\cite{noisy:new,prabha:limits} and does not lead to
	any error terms. As shown in Section~\ref{sec:URbounds}, any QC-extractor yields an entropic uncertainty relation with quantum side information. 
	For the case of 3 MUBs per qubit as in the protocol above, this uncertainty relation (see Table~\ref{tab:URsummary} with $\delta=0$) is given by
	\begin{align}\label{eq:entBound}
		H_{\min}^{\eps}(X|BM\Theta)_{{\rho}}
		\geq n\cdot\left(\log(3)-1\right) + \min\left\{0,\hmin(A|BM)_{{\rho}}-\kappa\right\} - \xi - 1\ .
	\end{align}
	Note that any decoding attack of Bob is absorbed into the operational interpretation of the min-entropy.
	As outlined above, it also follows from the operational interpretation of the min-entropy that for any encoding $\cE$ and decoding $\cD$ attack of Bob
	\begin{align}
		\hmin(A|BM)_{\rho} \geq - \log 2^n F_c(\cD \circ (\cF \otimes \cI_{\msg}) \circ \cE)\ .
	\end{align}
	Together with Equation~\eqref{eq:entBound} this yields our claim.
\end{proof}

Second, we consider a case of practical interest, i.e., channels of the form
$\cF = \cN^{\otimes N}$. 
Let us first establish some basic limits to security in this case.
Note that for rates $R \leq Q_{\rightarrow}(\cN)$, we have from Equation~\eqref{eq:defcap} that information \emph{can} be sent
reliably. That is, cheating Bob \emph{is} able to store the transmitted qubits perfectly whenever Alice sends less than $n = RN \leq Q_{\rightarrow}(\cN) N$ qubits. 
Note that 
\begin{align}
	R = \frac{1}{\nu}\ , 
\end{align}
and thus in terms of the storage rate $\nu$ this condition 
reads $1 \leq Q_{\rightarrow}(\cN) \cdot \nu$. Clearly, security cannot be obtained in this case.


\subsubsection{Strong converse parameter}

But what happens for $R > Q_{\rightarrow}(\cN)$? A \emph{weak converse} for the quantum capacity states that for any encoding $\cE$ and decoding scheme $\cD$, the channel fidelity is bounded away from $1$.
A \emph{strong converse} states that for any encoding and decoding scheme
\begin{align}\label{eq:chanConv}
	F_c(\cD \circ (\cF \otimes \cI_{\msg}) \circ \cE) \leq 2^{- \gamma^Q(\cN,R) \cdot N}\ ,
\end{align}
where $\gamma^Q(\cN,R) > 0$ is the \emph{strong converse parameter} of the channel $\cN$ at rate $R$. We are now ready to make a formal statement of security.
For this special case we obtain the following corollary by combining Theorem~\ref{thm:para}, Equation~\eqref{eq:chanConv} and $N = \nu \cdot n$.

\begin{corollary}\label{cor:para}
	Let Bob's storage device be of the form $\cF = \cN^{\otimes \nu n}$ with $\nu \cdot Q_{\rightarrow}(\cN) < 1$ and either $\nu \cdot \gamma^Q(\cN,1/\nu) > 2 - \log(3)$ or $\nu \cdot \gamma^Q(\cN,1/\nu) 
	< 1 + \kappa/n$.
	For any choice of constant parameters $\eps,\delta' > 0$, 
	Protocol 1 implements $(n,\lambda,\eps,1/3)$-WSE with 
	\begin{align}
		\lambda = \log(3) - 1 - \max\left\{0,\nu \cdot \gamma^Q(\cN,1/\nu) - 1 - \frac{\kappa}{n}\right\}
		- \frac{1}{n}\left(\xi + 1 \right)\ ,
	\end{align}
	where $\kappa = \log\left(2/\delta'^2+1\right)$ and $\xi = \log\left(1/\left(\eps^{2}/2-\delta'\right)^2\right)$.
\end{corollary}

Note that at first glance, the condition $\nu \cdot \gamma^Q(\cN,R) > 2 - \log(3)$ seems to favor large $\nu$. However, note that $\gamma^Q$ will be larger if the rate $R=1/\nu$ at which
we send information is higher. An illustrative example is provided below.


Given $Q_{\rightarrow}(\cN)$ and $\gamma^Q(\cN,R)$ we can thus in principle apply the theorem above to evaluate security parameters for any choice of $\cN$. 
Yet, it should be emphasized that determining the quantum capacity of a channel is in general a very hard problem.
Indeed, with the exception of so-called degradable channels, determining the quantum capacity of even rather innocent looking channels forms an elusive problem (see e.g.~\cite{ouyang} and references therin).
For example, even for the depolarizing channel which either outputs the original state with some probability $r$, or otherwise replaces it with the fully
mixed state, mere bounds on the quantum capacity are known.
Since a strong converse implies a sharp bound for information transmission, the existence of a strong converse for rates $R$ above a certain
threshold places a bound on the capacity. Hence, it is not surprising that determining the strong converse parameter for a channel $\cN$ when sending
information at a rate $R$ poses a challenge. For a long time it was only known that such a parameter exists for $R > C_E(\cN)/2$, where $C_E(\cN)$ is the \emph{classical}
entanglement assisted capacity of $\cN$. Indeed, the first further result was obtained only very recently by showing $\gamma^Q(\cN,R) > 0$ for $R > E_C(\cN)$, where $E_C(\cN) \geq Q_{\rightarrow}(\cN)$ is the \emph{entanglement cost} 
of $\cN$, capturing aspects of how well \emph{quantum} rather than \emph{classical} information can be transmitted through $\cN$~\cite{entCost}.


\subsubsection{Example: bounded storage}\label{sec:why}

Yet, to get some intuition about the parameters above, let us now consider the example of bounded, noise-free, storage. The quantum capacity of the one qubit identity channel
$\cN = \cI_2$ is simply $Q_{\rightarrow}(\cI_2) = 1$. A strong converse
is easy to obtain~\cite{entCost}. For completeness, we here provide a simple argument with slightly better parameters in the case of
classical forward communication.

\begin{lemma}
	The strong converse parameter of the one qubit identity channel obeys $\gamma^Q(\cI_2,R) = R-1 > 0$. 
\end{lemma}

\begin{proof}
Consider a decomposition of the encoding and decoding map in terms of their Kraus operators as $\cE(\rho) = \sum_j E_j \rho E_j^\dagger$ and
$\cD(\rho) = \sum_{k,m} \hat{D}_{k,m} \rho \hat{D}_{k,m}^\dagger$ where $\hat{D}_{k,m} = D_{k,m} \otimes \proj{m}$. Note 
that wlog the latter has this form since it is processing classical forward communication on $\msg$.
Let $\Pi_{k,m}$ denote the projector onto the subspace that $\hat{D}_{k,m}$ maps to.
We can now bound
\begin{align}
	F_c(\cD \circ (\cI_2^{\otimes N} \otimes \cI_{\msg}) \circ \cE) &= \sum_{jkm} \left|\tr\left[\hat{D}_{k,m}E_j\left(\frac{\id}{2^{NR}}\right)\right]\right|^2\\
	&\leq \sum_{jkm} \tr\left[\hat{D}_{k,m}E_j \left(\frac{\id}{2^{NR}}\right) E_j^\dagger \hat{D}_{k,m}^\dagger\right] \tr\left[\Pi_{k,m}\left(\frac{\id}{2^{NR}}\right)\right]\\
	&\leq 2^{-(R-1)N} \tr\left[\cD\circ\cE\left(\frac{\id}{2^{NR}}\right)\right]\\
	&= 2^{-(R-1)N}\ ,
\end{align}
where the first equality is a standard rewriting~\cite{NieChu00Book}, the second is given by the Cauchy-Schwarz inequality, and the last equality is given by the fact
that $\cD$ and $\cE$ are trace preserving.
\end{proof}

Plugging this strong converse parameter into Theorem~\ref{thm:para} and noting that $R = 1/\nu$ we obtain the following.

\begin{corollary}
	Let Bob's storage device be of the form $\cF = \cI_2^{\otimes \nu n}$ with $\nu < \log(3) - 1 \approx 0.585$.
	For any choice of constant parameters $\eps,\delta' > 0$, Protocol 1 implements $(n,\lambda,\eps,1/3)$-WSE with 
	\begin{align}
		\lambda = 
		(\log(3) - 1) - \nu -
		\frac{1}{n}\left(\kappa + \xi + 1 \right)\ ,
	\end{align}
	where $\kappa = \log\left(2/\delta'^2+1\right)$ and $\xi = \log\left(1/\left(\eps^{2}/2-\delta'\right)^2\right)$.
\end{corollary}

We note that for the case of bounded storage in an independent and identically distributed asymptotic setting, that is $\cF = \cI_2^{\otimes \nu n}$ with $n\rightarrow\infty$, the parameters obtained here are slightly worse than what was obtained in~\cite{prabha:limits}, where security was shown to be possible for $\nu < 2/3$ instead of $\nu \lesssim 0.585$. This is due to the fact that the lower bound $0.585$ in our uncertainty relation stems from an expression involving the \emph{collision} entropy (see Appendix~\ref{app:coll} for the definition) rather than the Shannon entropy. We emphasize however, that due to finite size effects our bound is still better in the practically relevant regime of $n \lesssim10^{6}$ (for the same security parameters).


\section{Discussion and Outlook}

Motivated by the problem of using physical resources to extract true classical randomness,
we introduced the concept of quantum-to-classical randomness extractors. We emphasize that these QC-extractors also work against quantum side information.
We showed that for a QC-extractor to distill randomness from a quantum state $\rho_{AE}$, the relevant quantity to bound is the conditional min-entropy $H_{\min}(A|E)_{\rho}$. 
This is in formal analogy with classical-to-classical extractors, in which case the relevant quantity is $H_{\min}(X|E)_{\rho}$. 

We proceeded by showing various properties of QC-extractors and giving several examples for QC-extractors. In this context, it is illustrative to compare our results about QC-extractors with CC-extractors (holding against quantum side information as well). This is done in Table~\ref{tab:ext-summary}.

\begin{table}[ht]
\centering
\begin{tabular}{|l|l|c|c|}
	\hline
	\multicolumn{2}{|c|}{} & CC-extractors & QC-extractors \\
	\hline
\multirow{3}{*}{Seed} & Lower bound &  $\log(n-k) + 2 \log(1/\eps)$ \ \cite{RTS00}  & $\log(1/\eps)$ \\[3mm]
  & \multirow{2}{*}{Upper bounds} & $\log(n-k) + 2 \log(1/\eps)$ (NE)  & $m + \log n + 4 \log(1/\eps)$ \ [Th \ref{thm:small-decoupling-set}] \ (NE)  \\
  &  & $c \cdot \log(n/\eps)$ \ \cite{GUV09}   & $3n$ \ [Th \ref{thm:full-set-mub}]  \\ \hline
\multirow{2}{*}{Output} & Upper bound & $k-2\log(1/\eps)$ \ \cite{RTS00} & $n+\entHmin^{\sqrt{\eps}}(A|E)$ \ [Pr \ref{lem:optimality}] \\
  & Lower bound & $k - 2 \log(1/\eps)$ \ \cite{ILL89,RK05,TSSR10}  & $n + k - 2\log(1/\eps)$ \ [Th \ref{thm:full-set-mub}]    \\ \hline
\end{tabular}
\caption{Known bounds on the seed size and output size in terms of (qu)bits for different kinds of $(k, \eps)$-randomness extractors. $n$ refers to the number of input (qu)bits, $m$ the number of output (qu)bits and $k$ the min-entropy of the input $\entHmin(A|E)$. Note that for QC-extractors, $k$ can be as small as $-n$. Additive absolute constants are omitted. The symbol (NE) denotes non-explicit constructions.}
\label{tab:ext-summary}
\end{table}

It is eye-catching that there is a vast difference between the upper and lower bounds for the seed size of QC-extractors. 
We were only able to show the existence of QC-extractors with seed length roughly the output size $m$, but we believe that it should be possible to find QC-extractors with much smaller seeds, say $O(\mathrm{polylog}(n))$ bits long, where $n$ is the input size. However, completely different techniques might be needed to address this question. 

It is interesting to note that our results do indeed lend further justification to use Bell tests to certify randomness created by measuring a quantum system~\cite{roger:thesis,pironio:nature,massar:recent,Colbeck11}.
Note that for a tripartite pure state $\rho_{ABE}$ where we want to create classical randomness by means of QC-extractors on $A$, we have to find a lower bound on $\hmin(A|E)_{\rho}$. But by the duality relation for min/max-entropies we have $\hmin(A|B)_{\rho}=-{H}_{\max}(A|B)_{\rho}$~\cite{Tomamichel09}, where the latter denotes the max-entropy as introduced~\cite{krs:operational}. 
Since ${H}_{\max}(A|B)_{\rho}$ is again a measure for the entanglement between $A$ and $B$, one basically only has to do entanglement witnessing (e.g.,~Bell tests consuming part of the state)
to ensure that the QC-extractor method can work (i.e.~that $\hmin(A|E)_{\rho}$ is large enough). Note that any method to certify such an estimate would do 
and we could also use different measurements during the estimation process and the final extraction step. 
It would be interesting to know, if by using a particular QC-extractor, one can gain more randomness than in~\cite{roger:thesis,pironio:nature,massar:recent,Colbeck11}.
In~\cite{massar:recent}, it was also remarked that if we want to extract randomness from $A$ \emph{and} $B$, then it is not necessary for the joint state across $A$ and $B$ to be maximally entangled. Note that this is indeed intuitive as the amount of extractable randomness in this case is determined by $\hmin(AB|E)$ together.

As the first application, we showed that every QC-extractor gives rise to entropic uncertainty relations with quantum side information for the von Neumann (Shannon) entropy and the min-entropy.  Here the seed size translates into the number  of measurements in the uncertainty relation. Since it is in general difficult to obtain uncertainty relations for a small set of measurements (except for the special case of two), finding QC-extractors with a small seed size is also worth pursuing from the point of view of uncertainty relations. 

As the second application, we used the bitwise QC-extractor from Section~\ref{sec:singlequdit} to show that the security in the noisy storage model can be related to the strong converse rate of the quantum 
storage; a problem that attracted quite some attention over the last few years. Here one can also see the usefulness of \emph{bitwise} QC-extractors for quantum cryptography. 
Indeed, any bitwise QC-extractor would yield a protocol for weak string erasure.
Bitwise measurements have a very simple structure, and hence are implementable with current technology. 
In that respect, it would be interesting to see if a similar QC-extractor can also be proven for only two (complementary) measurements per qubit. This would give a protocol for
weak string erasure using BB84 bases as in~\cite{noisy:new}.

We expect that QC-extractors will have many more applications in quantum cryptography, e.g.,~quantum key distribution.
One possible interesting application could be to prove the security of oblivious transfer when purifying the protocol of~\cite{noisy:new}. Yet, it would require additional concepts of \lq entanglement sampling\rq~which still elude us.


\section*{Acknowledgments}

We gratefully acknowledge useful discussions with Matthias Christandl, Patrick Hayden, Robert K{\"o}nig, Joseph Renes, Pranab Sen, Oleg Szehr, Marco Tomamichel, J\"{u}rg Wullschleger and Mark Wilde. MB is supported by the Swiss National Science Foundation (grant PP00P2-128455), the German Science Foundation (grants CH 843/1-1 and CH 843/2-1), and the Swiss National Centre of Competence in Research 'Quantum Science and Technology'. OF is supported by CIFAR, NSERC and ONR grant No. N000140811249. MB and OF thank the Center for Quantum Technologies, Singapore, for hosting them while part of this work was done. SW thanks the National Research Foundation and Ministry of Education, Singapore.


\appendix

\section{Properties of Smooth Entropy Measures}

\subsection{Collision entropy and alternative smooth entropies}\label{app:coll}

For technical reasons, we need some more entropic quantities. We start with the quantum conditional collision entropy. For a state $\rho_{AB} \in \states(AB)$ relative to a state $\sigma_B \in \states(B)$, it is defined as
\begin{align}\label{eq:coll}
H_2(A|B)_{\rho|\sigma}=-\log\tr\left[(\id_A \ox \sigma_B^{-1/4}) \rho_{AB} (\id_A \ox \sigma_B^{-1/4})\right]^{2}\ ,
\end{align}
where the inverses are generalized inverses.\footnote{For $M_{A}\in\cL(A)$, $M_{A}^{-1}$ is a generalized inverse of $M_{A}$ if $M_{A}M_{A}^{-1}=M_{A}^{-1}M_{A}=\mathrm{supp}(M_{A})=\mathrm{supp}(M_{A}^{-1})$, where $\mathrm{supp}(.)$ denotes the support.} Next we introduce the following alternative smooth conditional min-entropy. For a state $\rho_{AB} \in \states(AB)$ it is defined as
\begin{align}
\entHmin^{\eps}(A|B)_{\rho|\rho}=\max_{\tilde{\rho}_{AB} \in \epsball(\rho_{AB})} \entHmin(A|B)_{\tilde{\rho}|\tilde{\rho}}\ .
\end{align}
We will also need the conditional max-entropy
\begin{align}\label{eq:maxentropy}
H_{\max}(A|B)_{\rho}=\max_{\sigma_{B}\in\cS(B)}\log F(\rho_{AB},\id_{A}\otimes\sigma_{B})^{2}\ ,
\end{align}
and its smooth version
\begin{align}\label{eq:smoothmaxentropy}
H_{\max}^{\eps}(A|B)_{\rho}=\min_{\tilde{\rho}_{AB}\in\cB^{\eps}(\rho_{AB})}H_{\max}(A|B)_{\tilde{\rho}}\ .
\end{align}

The following lemma relates the collision and the min-entropy.

\begin{lemma}\label{lem:hmin-h2-rhorho}
Let $\rho_{AB}\in\cS_{\leq}(AB)$ and $\sigma_{B}\in\cS(B)$ with $\mathrm{supp}(\rho_{AB})\subseteq\id_{A}\otimes\mathrm{supp}(\sigma_{B})$, where $\mathrm{supp}(.)$ denotes the support. Then
\begin{align}
\hmin(A|B)_{\rho|\sigma}\leq H_{2}(A|B)_{\rho|\sigma}\ .
\end{align}
\label{collision}
\end{lemma}
\begin{proof}
We have $\mathrm{supp}(\rho_{AB})\subseteq\id_{A}\otimes\mathrm{supp}(\rho_{B})$ and hence by~\cite[Lemma B.2]{Berta09_2}
\begin{align}
\hmin(A|B)_{\rho|\sigma}=-\log\max_{\omega_{AB}\in\cS(AB)}\tr\left[\omega_{AB}\left(\id_{A}\otimes\sigma_{B}^{-1/2}\right)\rho_{AB}\left(\id_{A}\otimes\sigma_{B}^{-1/2}\right)\right]\ ,
\end{align}
where the inverses are generalized inverses. But for $\hat{\rho}_{AB}=\frac{\rho_{AB}}{\tr\left[\rho_{AB}\right]}\in\cS(AB)$ we have,
\begin{align}
H_{2}(A|B)_{\rho|\sigma}&=-\log\tr\left[\rho_{AB}\left(\id_{A}\otimes\sigma_{B}^{-1/2}\right)\rho_{AB}\left(\id_{A}\otimes\sigma_{B}^{-1/2}\right)\right]\\
&=-\log\tr\left[\rho_{AB}\right]-\log\tr\left[\hat{\rho}_{AB}\left(\id_{A}\otimes\sigma_{B}^{-1/2}\right)\rho_{AB}\left(\id_{A}\otimes\sigma_{B}^{-1/2}\right)\right]\\
&\geq-\log\max_{\omega_{AB}\in\cS(AB)}\tr\left[\omega_{AB}\left(\id_{A}\otimes\sigma_{B}^{-1/2}\right)\rho_{AB}\left(\id_{A}\otimes\sigma_{B}^{-1/2}\right)\right]=\hmin(A|B)_{\rho|\sigma}\ .
\end{align}
\end{proof}

Finally, we also need a relation between the standard min-entropy, and the alternative definition from above.

\begin{lemma}\cite[Lemma 18]{TSSR10}
\label{lem:hmin-rhorho}
Let $\eps'\geq0$, $\eps'>0$, and $\rho_{AB} \in \cS(AB)$. Then
\begin{align}
H^{\eps}_{\min}(A|B)_{\rho} - \log\left(\frac{2}{\eps'^2} + \frac{1}{1-\eps}\right) \leq H^{\eps+\eps'}_{\min}(A|B)_{\rho|\rho} \leq H^{\eps+\eps'}_{\min}(A|B)_{\rho}\ .
\end{align}
\end{lemma}


\subsection{Chain Rules}

The smooth conditional min- and max-entropy fulfill a duality relation.

\begin{lemma}\cite{Tomamichel09}\label{lem:duality}
Let $\rho_{AB}\in\cS(AB)$, $\eps\geq0$, and $\rho_{ABC}$ be an arbitrary purification of $\rho_{AB}$. Then
\begin{align}
H_{\max}^{\eps}(A|B)_{\rho}=-H_{\min}^{\eps}(A|C)_{\rho}\ .
\end{align}
\end{lemma}

The following shows that the min-entropy can not increase too much by a measurement on the first system.

\begin{lemma}\label{lem:qcvscc}
Let $\rho_{AB} \in \states(AB)$, $\eps\geq0$, and $\{P_{x}\}_{x=1}^{|X|}$ be a projective rank-one measurement on $A$. Then
\begin{align}
H_{\min}^{\eps}(X|B)_{\rho}\leq H_{\min}^{\eps}(A|B)_{\rho}+\log|X|\ .
\end{align}
\end{lemma}

\begin{proof}
Let $V_{A\rightarrow XX'}$ be an isometric purification of $\{P_{x}\}$ and $\rho_{XX'BB'}$ a purification of $\rho_{XX'B}=V\rho_{AB}V^{\dagger}$. By the invariance of the min-entropy under local isometries~\cite[Lemma 13]{Tomamichel09} and the duality between the min- and max-entropy (Lemma~\ref{lem:duality}), the proposition becomes equivalent to
\begin{align}
H_{\max}^{\eps}(XX'|B')_{\rho}\leq H_{\max}^{\eps}(X|X'B')_{\rho}+\log|X|\ .
\end{align}
For $\hat{\rho}_{XX'B'}\in\cB^{\eps}(\rho_{XX'B'})$ and $\hat{\sigma}_{X'B'}\in\cS(X'B')$ such that
\begin{align}
H_{\max}^{\eps}(X|X'B')_{\rho}=\log F(\hat{\rho}_{XX'B},\id_{X}\otimes\hat{\sigma}_{X'B'})^{2}\ ,
\end{align}
as well as $\bar{\rho}_{XX'B'}\in\cB^{\eps}(\rho_{XX'B'})$ and $\bar{\sigma}_{B}\in\cS(B)$ such that
\begin{align}
H_{\max}^{\eps}(XX'|B')_{\rho}=\log F(\bar{\rho}_{XX'B},\id_{XX'}\otimes\bar{\sigma}_{B'})^{2}\ ,
\end{align}
the claim follows by the definition of the max-entropy (Equation~\eqref{eq:maxentropy}-\eqref{eq:smoothmaxentropy}) together with the observation
\begin{align}
H_{\max}^{\eps}(XX'|B')_{\rho}\leq\log \left( |X|\cdot F(\hat{\rho}_{XX'B'},\id_{X}\otimes\frac{\id_{X'}}{|X|}\otimes\bar{\sigma}_{B'})^{2} \right) \leq\log F(\hat{\rho}_{XX'B'},\id_{X}\otimes\hat{\sigma}_{X'B'})^{2}+\log|X|\ .
\end{align}
\end{proof}

\subsection{Asymptotic behavior}

The von Neumann entropy can be seen as a special case of the smooth min-entropy. The underlying technical statement that makes this precise, is the asymptotic equipartition property (AEP) for the smooth conditional min-entropy.

\begin{lemma}\cite[Remark 10]{Tomamichel08}\label{lem:aep}
Let $\rho_{AB}\in\cS(AB)$, $\eps>0$, and $n\geq2\left(1-\eps^{2}\right)$. Then,
\begin{align}
\frac{1}{n}\hmin^{\eps}(A|B)_{\rho^{\otimes n}|\rho^{\otimes n}}\geq H(A|B)_{\rho}-\frac{4\sqrt{1-2\log\eps}\left(2+\frac{\log|A|}{2}\right)}{\sqrt{n}}\ .
\end{align}
\end{lemma}


\section{Technical Lemmata}\label{app:technical}

Throughout, we will need a number of technical results and definitions, summarized here for convenience. In the following we state all results in our own notation and only as general as we need them (which may result in a simplification compared to the given references). We start with a general decoupling result about exact unitary 2-designs.

\begin{lemma}\cite[Theorem 3.7]{frederic:decoupling}\label{lem:2design}
Let $A=A_{1}A_{2}$, and consider the map $\cT_{A\rightarrow A_{1}}$ as defined in Equation~\eqref{eq:meas-map}. Then, if $\left\{U_1, \dots, U_{L}\right\}$ defines an exact unitary 2-design (Definition~\ref{def:two-design}), we have for $\delta\geq0$,
\begin{align}\label{eq:2design}
\frac{1}{L}\sum_{i=1}^{L}\left\|\cT_{A\rightarrow A_{1}}(U_{i}\rho_{AE}U_{i}^{\dagger})-\frac{\id_{A_{1}}}{|A_{1}|}\otimes\rho_{E}\right\|_{1}\leq\sqrt{\frac{|A_{1}|}{|A|}2^{-\hmin^{\delta}(A|E)_{\rho}}}+2\delta\ .
\end{align}
\end{lemma}



The full set of MUBs generates a complex projective 2-design.

\begin{lemma}\cite{KR05}\label{lem:mub-design}
Let $\left\{U_{1},\dots,U_{|A|+1}\right\}$ define a full set of mutually unbiased bases of $A$. Then
\begin{align}
\frac{1}{|A|(|A|+1)} \sum_{i=1}^{|A|+1} \sum_{a \in [d]} (U_i \proj{a} U_i^{\dagger})^{\ox 2} = \frac{2 \Pi^{\sym}}{|A|(|A|+1)},
\end{align}
where $\Pi^{\sym}$ is the projector onto the symmetric subspace spanned by the vectors $\ket{aa'} + \ket{a'a}$ for $a,a' \in [A]$ .
\end{lemma}

The following well known \lq swap trick\rq~is used to prove decoupling statements.

\begin{lemma}\label{lem:swap-trick}
Let $M,N\in\cL(A)$. Then,
\begin{align}
\tr [ MN ] = \tr [(M \ox N)F],
\end{align}
where $F = \sum_{aa'} \ket{aa'}\bra{a'a}$ is the swap operator.
\end{lemma}

The following is called operator Chernoff bound.

\begin{lemma}\cite[Theorem 19]{AW02}
\label{lem:operator-chernoff}
Let $X_1, \dots, X_L$ be iid random variables and $0 \leq X_i \leq \id$, $\ex{X_i} = \Gamma \geq \alpha \id$. Then
\begin{align}
\prob{\frac{1}{L} \sum_{i=1}^L X_i \leq (1+\eta) \Gamma} \geq 1 - d \exp \left(- \frac{L \eta^2 \alpha}{2 \ln 2} \right).
\end{align}
\end{lemma}


\section{Proofs of QC-Extractors}\label{app:proofs}

In this section, we provide the full proofs of our claims regarding QC-extractors. In the proofs we need the Hilbert-Schmidt norm, given by $\|\rho\|_{2}=\sqrt{\tr\left[\rho^{\dagger}\rho\right]}$.

\newtheorem*{thm-small-decoupling-set}{Theorem \ref{thm:small-decoupling-set}}
\begin{thm-small-decoupling-set}
Let $A=A_{1}A_{2}$ with $n = \log |A|$ and $\cT_{A\rightarrow A_{1}}$ be the measurement map defined in Equation~\eqref{eq:meas-map}. Let $\eps > 0$, $c$ be a sufficiently large constant, and
\begin{align}
\log |A_1| \leq n + k - 4 \log (1/\eps) - c \qquad \text{ as well as } \qquad \log L \geq \log |A_1| + \log n + 4 \log(1/\eps) + c \ .
\end{align}
Then, choosing $L$ unitaries $\left\{U_1, \dots, U_L\right\}$ independently according to the Haar measure defines a $(k, \eps)$-QC-extractor with high probability (see Equation~\eqref{eq:prob-decoupling-set} for a precise bound).
\end{thm-small-decoupling-set}

\begin{proof}
We use one-shot decoupling techniques as developed in~\cite{Berta08,frederic:decoupling,Wullschleger08,Szehr11,szehr:designs}. Let $U$ be a unitary on $A$. Using the Hoelder-type inequality (see e.g.~\cite{Bhatia97})
\begin{align}
\|\alpha\beta\gamma\|_{1}\leq\||\alpha|^{r}\|_{1}^{1/r}\||\beta|^{s}\|_{1}^{1/s}\||\gamma|^{t}\|_{1}^{1/t}
\end{align}
with $r=t=4$, $s=2$, and $\alpha=\gamma=(\id_{A_{1}}\otimes\rho_{E})^{1/4}$, $\beta=(\id_{A_{1}}\otimes\rho_{E})^{-1/4}\left(\cT(U \rho_{AE} U^{\dagger})-\frac{\id_{A_{1}}}{|A_{1}|} \ox \rho_E\right)(\id_{A_{1}}\otimes\rho_{E})^{-1/4}$, we get that\footnote{The inverses are generalized inverses.}
\begin{align}
\left\| \cT\left(U\rho_{AE}U^{\dagger}\right) - \frac{\id_{A_{1}}}{|A_{1}|} \ox \rho_E \right\|_1&
\leq|A_{1}|^{1/4}\sqrt{\tr\left[\left(\id_{A_{1}}\otimes\rho_{E}\right)^{-1/4}\left(\cT\left(U\rho_{AE}U^{\dagger}\right) - \frac{\id_{A_{1}}}{|A_{1}|} \ox \rho_E\right)\left(\id_{A_{1}}\otimes\rho_{E}\right)^{-1/4}\right]^{2}}|A_{1}|^{1/4}\\
&=|A_{1}|^{1/2}\left\|\left(\id_{A_{1}}\otimes\rho_{E}\right)^{-1/4}\left(\cT\left(U\rho_{AE}U^{\dagger}\right) - \frac{\id_{A_{1}}}{|A_{1}|} \ox \rho_E\right)\left(\id_{A_{1}}\otimes\rho_{E}\right)^{-1/4}\right\|_{2}\\
&=|A_{1}|^{1/2}\left\|\cT\left(U\tilde{\rho}_{AE}U^{\dagger}\right)-\frac{\id_{A_{1}}}{|A_{1}|}\otimes\tilde{\rho}_{E}\right\|_{2}\ ,
\end{align}
where $\tilde{\rho}_{AE}=(\id_{A}\otimes\rho_{E})^{-1/4}\rho_{AE}(\id_{A}\otimes\rho_{E})^{-1/4}$. Together with the concavity of the square root function, this implies
\begin{align}
\frac{1}{L} \sum_{i=1}^{L}\left\| \cT\left(U_i \rho_{AE}U_i^{\dagger}\right) - \frac{\id_{A_{1}}}{|A_{1}|} \ox \rho_E \right\|_1&
\leq\sqrt{\frac{1}{L} \sum_{i=1}^{L}\left\| \cT\left(U_i \rho_{AE}U_i^{\dagger}\right) - \frac{\id_{A_{1}}}{|A_{1}|} \ox \rho_E \right\|_1^{2}}\\
&\leq\sqrt{|A_{1}|\frac{1}{L} \sum_{i=1}^{L}\left\| \cT\left(U_i \tilde{\rho}_{AE}U_i^{\dagger}\right) - \frac{\id_{A_{1}}}{|A_{1}|} \ox \tilde{\rho}_E \right\|_2^{2}}\\
&=\sqrt{|A_{1}|\frac{1}{L} \sum_{i=1}^{L}\tr\left[\cT\left(U_i \tilde{\rho}_{AE}U_i^{\dagger}\right) - \frac{\id_{A_{1}}}{|A_{1}|} \ox \tilde{\rho}_E\right]^{2}}\ .\label{eq:concave}
\end{align}
We continue with
\begin{align}
&\frac{1}{L} \sum_{i=1}^{L}\tr\left[\cT\left(U_i \tilde{\rho}_{AE}U_i^{\dagger}\right) - \frac{\id_{A_{1}}}{|A_{1}|} \ox \tilde{\rho}_E\right]^{2}\\
&=\frac{1}{L} \sum_{i=1}^{L}\tr\left[\cT\left(U_i \tilde{\rho}_{AE} U_i^{\dagger}\right)\right]^{2}-2\tr\left[\cT\left(U_i \tilde{\rho}_{AE}U_i^{\dagger}\right)\left(\frac{\id_{A_{1}}}{|A_{1}|} \ox \tilde{\rho}_E\right)\right]+\tr\left[\frac{\id_{A_{1}}}{|A_{1}|} \ox \tilde{\rho}_E\right]^{2} \label{eq:expandsquare}
\end{align}
and first compute the cross term
\begin{align}
\tr\left[\cT\left(U_i \tilde{\rho}_{AE}U_i^{\dagger}\right)\left(\frac{\id_{A_{1}}}{|A_{1}|} \ox \tilde{\rho}_E\right)\right]&
=\frac{1}{|A_{1}|}\tr\left[\tr_{A_{1}}\left[\cT\left(U_i \tilde{\rho}_{AE}U_i^{\dagger}\right)\left(\id_{A_{1}} \ox \tilde{\rho}_E\right)\right]\right]\\
&=\frac{1}{|A_{1}|}\tr\left[\tilde{\rho}_{E}^{2}\right]\ .
\end{align}
Going back to Equation~\eqref{eq:expandsquare}, we obtain
\begin{align}
\frac{1}{L} \sum_{i=1}^{L}\tr\left[\cT\left(U_i \tilde{\rho}_{AE}U_i^{\dagger}\right) - \frac{\id_{A_{1}}}{|A_{1}|} \ox \tilde{\rho}_E\right]^{2}=\frac{1}{L} \sum_{i=1}^{L}\tr\left[\cT\left(U_i \tilde{\rho}_{AE}U_i^{\dagger}\right)\right]^{2}-\frac{1}{|A_{1}|}\tr\left[\tilde{\rho}_{E}^{2}\right]\ .\label{eq:middle}
\end{align}
We now compute the first term using the \lq swap trick\rq~(Lemma~\ref{lem:swap-trick})
\begin{align}
\tr\left[\cT(U\tilde{\rho}_{AE}U^{\dagger}) \right]^2
&= \tr\left[ \sum_{a_{1}a_{2}} \bra{a_{1}a_{2}} U\tilde{\rho}_{AE}U^{\dagger} \ket{a_{1}a_{2}} \proj{a_{1}} \right]^2 \\
&= \tr\left[ \sum_{a_{1}a_{2}a_{1}'a_{2}'}\bra{a_{1}a_{2}a_{1}'a_{2}'}U^{\ox 2} \tilde{\rho}^{\ox 2}_{AE} (U^{\ox 2})^{\dagger} \ket{a_{1}a_{2}a_{1}'a_{2}'} \proj{a_{1}a_{1}'}\left(F_{A_{1}A_{1}'} \otimes F_{EE'}\right)\right] \\
&= \sum_{a_{1}a_{2}a_{1}'a_{2}'} \tr\left[ \tilde{\rho}^{\ox 2}_{AE} (U^{\ox 2})^{\dagger} \ket{a_{1}a_{2}a_{1}'a_{2}'} \bra{a_{1}a_{1}'}\left(F_{A_{1}A_{1}'} \otimes F_{EE'}\right) \ket{a_{1}a_{1}'} \bra{a_{1}a_{2}a_{1}'a_{2}'} U^{\ox 2} \right]\ .
\end{align}
Taking the average over the set $\left\{U_1, \dots, U_L\right\}$, we get
\begin{align}
\frac{1}{L} \sum_{i=1}^{L}\tr\left[\cT\left(U_i \tilde{\rho}_{AE}U_i^{\dagger}\right)\right]^{2}
&= \sum_{a_{1}a_{2}a_{1}'a_{2}'} \tr\left[ \tilde{\rho}^{\ox 2}_{AE} \frac{1}{L} \sum_{i=1}^{L}\left\{\left(U_i^{\ox 2}\right)^{\dagger} \ket{a_{1}a_{2}a_{1}'a_{2}'} \bra{a_{1}a_{1}'}F_{A_{1}A_{1}'}\ket{a_{1}a_{1}'} \bra{a_{1}a_{2}a_{1}'a_{2}'}U_i^{\ox 2}\right\}  \otimes F_{EE'}\right]\\
&=  \tr\left[ \tilde{\rho}^{\ox 2}_{AE} \frac{1}{L} \sum_{i=1}^{L}\left\{\left(U_i^{\dagger}\right)^{\ox 2} \sum_{a_{1}a_{2}a_{2}', a_1' = a_1} \ket{a_{1}a_{2}a_{1}'a_{2}'} \bra{a_{1}a_{2}a_{1}'a_{2}'}U_i^{\ox 2}\right\}  \otimes F_{EE'}\right]\ .\label{eq:average}
\end{align}
Using for example~\cite[Lemma 3.4]{Wullschleger08}, if $U$ is distributed according to the Haar measure on the group of unitaries acting on $A$, then
\begin{align}
\exc{U}{\left(U^{\dagger}\right)^{\ox 2} \sum_{a_{1}a_{2}a_{2}'} \ket{a_{1}a_{2}a_{1}a_{2}'} \bra{a_{1}a_{2}a_{1}a_{2}'}U^{\ox 2}}
&= \left( \frac{|A||A_{2}| - 1}{|A|^{2} -1} \right) \id_{AA'} + \frac{|A|-|A_{2}|}{|A|^{2}-1}F_{AA'}\equiv\Gamma_{AA'}\ .
\end{align}
Now we note that $\frac{|A||A_{2}| - 1}{|A|^{2} -1}\geq\frac{1}{2|A_{1}|}$, and apply an operator Chernoff bound (Lemma~\ref{lem:operator-chernoff}) to get
\begin{equation}
\label{eq:prob-decoupling-set}
\prob{\frac{1}{L} \sum_{i=1}^L  (U_i^{\dagger})^{\ox 2} \sum_{a_{1}a_{2}a_{2}'} \ket{a_{1}a_{2}a_{1}a_{2}'} \bra{a_{1}a_{2}a_{1}a_{2}'} U_i^{\ox 2} \leq (1+\eta) \Gamma } \geq1-|A|\exp\left(-\frac{L\eta^2}{|A_{1}|4 \ln 2} \right)\ .
\end{equation}
This shows that if $L\geq 2 \cdot 4\ln2\cdot|A_{1}|\log|A|/\eta^2$, the unitaries $U_1, \dots, U_L$ satisfy the above operator inequality with high probability. In the rest of the proof, we show that such unitaries define QC-extractors. Putting these unitaries in Equation~\eqref{eq:average}, we get
\begin{align}
\frac{1}{L} \sum_{i=1}^{L}\tr\left[\cT\left(U_i \tilde{\rho}_{AE} \left(U_i\right)^{\dagger}\right) \right]^2
\leq (1+\eta)\left( \frac{|A||A_{2}|-1}{|A|^{2}-1}\tr\left[ \tilde{\rho}_E^2 \right]+\frac{|A|-|A_{2}|}{|A|^{2}-1}\tr\left[ \tilde{\rho}_{AE}^2 \right] \right)\ .
\end{align}
Plugging this expression in Equation~\eqref{eq:middle} and then in Equation~\eqref{eq:concave}, we get
\begin{align}
\frac{1}{L} \sum_{i=1}^{L}\left\| \cT\left(U_i \rho_{AE} \left(U_i\right)^{\dagger}\right) - \frac{\id_{A_{1}}}{|A_{1}|} \ox \rho_E \right\|_1
&\leq \sqrt{(1+\eta)\left(\frac{|A|^2 - |A_{1}|}{|A|^2 - 1}\right) \tr\left[\tilde{\rho}_E^2\right]+(1+\eta)\left(\frac{|A_{1}||A|-|A|}{|A|^2 - 1} \right)\tr\left[ \tilde{\rho}_{AE}^2 \right]-\tr\left[\tilde{\rho}_{E}^{2}\right]}\\
&\leq \sqrt{\eta+(1+\eta)\frac{|A_{1}|}{|A|+1} \tr\left[\tilde{\rho}_{AE}^2\right]}\ ,
\end{align}
since $\tr \left[ \tilde{\rho}_{E}^2 \right] = \tr \left[\tr_A \left[\left(\id_A \ox {\rho}_E^{-1/4}\right) \rho_{AE} \left(\id_A \ox {\rho}_E^{-1/4}\right)\right]^{2}\right] = \tr \left[ \rho_E \right]=1$. By the definition of the conditional collision entropy (Equation~\eqref{eq:coll}) and Lemma~\ref{lem:hmin-h2-rhorho}, it follows that,
\begin{align}
\frac{1}{L} \sum_{i=1}^{L}\left\| \cT\left(U_i \rho_{AE} \left(U_i\right)^{\dagger}\right) - \frac{\id_{A_{1}}}{|A_{1}|} \ox \rho_E \right\|_1
&\leq\sqrt{\eta+(1+\eta)\frac{|A_{1}|}{|A|+1}2^{-H_{2}(A|E)_{\rho|\rho}}}\\
&\leq\sqrt{\eta+(1+\eta)\frac{|A_{1}|}{|A|+1}2^{-\entHmin(A|E)_{\rho|\rho}}}\ .
\label{eq:c25}
\end{align}
Now let $\rho'_{AE}\in\cB^{\delta+\delta'}(\rho_{AE})$ be such that $\entHmin^{\delta+\delta'}(A|E)_{\rho|\rho}=\entHmin(A|E)_{\rho'|\rho'}$. Since we have $\|\rho_{AE}'-\rho_{AE}\|_{1}\leq2(\delta+\delta')$ (by Equation~\eqref{eq:purifiedVStrace}), we know that by the (reverse) triangle inequality and the monotonicity of the trace distance,
\begin{align}
\left| \| \cT(U \rho_{AE} U^{\dagger}) - \frac{\id_{A_{1}}}{|A_{1}|} \ox \rho_E \|_1 - \| \cT(U \rho'_{AE} U^{\dagger}) - \frac{\id_{A_{1}}}{|A_{1}|} \ox \rho_E \|_1 \right| &\leq \| \cT(U \rho_{AE} U^{\dagger}) - \cT(U \rho'_{AE} U^{\dagger})  \|_1 \\
&\leq \| \rho'_{AE} - \rho_{AE} \|_1 \leq2(\delta+\delta')\ ,
\end{align}
and hence applying \eqref{eq:c25} to $\rho'_{AB}$, we get
\begin{align}
\frac{1}{L} \sum_{i=1}^{L}\left\| \cT\left(U_i \rho_{AE} \left(U_i\right)^{\dagger}\right) - \frac{\id_{A_{1}}}{|A_{1}|} \ox \rho_E \right\|_1
&\leq\sqrt{\eta+(1+\eta)\frac{|A_{1}|}{|A|+1}2^{-\entHmin^{\delta+\delta'}(A|E)_{\rho|\rho}}}+2(\delta+\delta')\ .
\end{align}
We then use Lemma~\ref{lem:hmin-rhorho} about the equivalence of the different conditional min-entropies to get
\begin{align}
\frac{1}{L} \sum_{i=1}^{L}\left\| \cT\left(U_i \rho_{AE} \left(U_i\right)^{\dagger}\right) - \frac{\id_{A_{1}}}{|A_{1}|} \ox \rho_E \right\|_1
&\leq\sqrt{\eta+(1+\eta)\frac{|A_{1}|}{|A|+1}2^{-\entHmin^{\delta}(A|E)_{\rho} + z}}+2(\delta+\delta')\ ,
\end{align}
with $z = \log(2/\delta'^2 + 1/(1-\delta))$. Setting $\eta = \eps^2/4$, $\delta = 0$, $\delta' = \eps/4$, and assuming $\log |A_1| \leq n+ k - 4 \log(1/\eps) - c$ with $k = \entHmin(A|E)_{\rho}$, we get for large enough $c$
\begin{align}
\sqrt{\eta+(1+\eta)\frac{|A_{1}|}{|A|+1}2^{-\entHmin(A|E)_{\rho} + z}}+\delta'&\leq\eps/2+\sqrt{\eps^2/4 + 2 \cdot 2^{k - 4 \log(1/\eps) - c -k + \log(8/\eps^2+1)}}\\
&\leq\eps/2+\sqrt{\eps^2/4 +\eps^2\cdot2^{1-c+4}}\leq\eps\ .
\end{align}
\end{proof}

\newtheorem*{thm-full-set-mub}{Theorem \ref{thm:full-set-mub}}
\begin{thm-full-set-mub}
Let $A=A_{1}A_{2}$ with $n = \log |A|$, $|A|$ a prime power, and consider the map $\cT_{A\rightarrow A_{1}}$ as defined in Equation~\eqref{eq:meas-map}. Then, if $\left\{U_1, \dots, U_{|A|+1}\right\}$ defines a full set of mutually unbiased bases, we have for $\delta\geq0$,
\begin{align}
 \frac{1}{|\cP|}  \frac{1}{|A|+1} \sum_{P \in \cP} \sum_{i=1}^{|A|+1} \left\| \cT_{A\rightarrow A_{1}}\left(PU_i \rho_{AE} \left(PU_i\right)^{\dagger}\right) - \frac{\id_{A_{1}}}{|A_{1}|} \ox \rho_E \right\|_1
\leq \sqrt{\frac{|A_{1}|}{|A|+1} 2^{-H^{\delta}_{\min}(A|E)_{\rho}}}+2\delta\ ,
\end{align}
where $\cP$ is a set of pair-wise independent permutation matrices. In particular, the set $\{PU_i : P \in \cP, i \in [|A|+1]\}$ defines a $(k, \eps)$-QC-extractor provided
\begin{align}
\log |A_1| \leq n + k - 2 \log(1/\eps)\ ,
\end{align}
and the number of unitaries is
\begin{align}
L = (|A|+1) |\cP| = (|A|+1)|A| (|A|-1)\ .
\end{align}
\end{thm-full-set-mub}

\begin{proof}
Let $\sigma_{E}\in\cS(E)$. Similarly as in the proof of Theorem~\ref{thm:small-decoupling-set}, but with the difference that now $\tilde{\rho}_{AE}=\left(\id_{A}\otimes\sigma_{E}\right)^{-1/4}\rho_{AB}\left(\id_{A}\otimes\sigma_{E}\right)^{-1/4}$, we get
\begin{align}
&\frac{1}{|\cP|}  \frac{1}{|A|+1} \sum_{P \in \cP} \sum_{i=1}^{|A|+1} \left\| \cT\left(PU_i \rho_{AE} \left(PU_i\right)^{\dagger}\right) - \frac{\id_{A_{1}}}{A_{1}} \ox \rho_E \right\|_1
\equiv\exc{P,i}{\left\| \cT\left(PU_i \rho_{AE} \left(PU_i\right)^{\dagger}\right) - \frac{\id_{A_{1}}}{A_{1}} \ox \rho_E \right\|_1}\\
&\leq\sqrt{|A_{1}|\sum_{a_{1}a_{2}a_{2}'}\tr\left[\tilde{\rho}_{AE}^{\otimes2}\exc{P,i}{\left(U_i^{\dagger}P^{\dagger}\right)^{\otimes2}\proj{a_{1}a_{2}a_{1}a_{2}'}\left(PU_i\right)^{\otimes2}}\otimes F_{EE'}\right]-\tr\left[\tilde{\rho}_{E}^{2}\right]}\ .\label{eq:start}
\end{align}
We handle the case $a_{2}=a_{2}'$ and the case $a_{2}\neq a_{2}'$ differently. When $a_{2}=a_{2}'$, we have $(U_i^{\dagger})^{\ox 2} \ket{aa} \bra{aa} U_i^{\ox 2} = (U_i^{\dagger}\ket{a} \bra{a} U_i)^{\ox 2}$, where $a = P^{-1}(a_{1}a_{2})$. As $\{U_1,\ldots,U_{|A|+1}\}$ form a full set of mutually unbiased bases, the vectors $\{U_i \ket{a}\}_{i,a}$ define a complex projective 2-design (Lemma~\ref{lem:mub-design}), and we get
\begin{align}
\sum_{a_{1}a_{2}, a_{2}'=a_{2}} \exc{P,i}{ \left(U_i^{\dagger}P^{\dagger}\right)^{\ox 2} \ket{a_{1}a_{2}a_{1}a_{2}'} \bra{a_{1}a_{2}a_{1}a_{2}'}\left(PU\right)^{\ox 2}}&
=\sum_{a} \exc{i}{\left(U_i^{\dagger}\right)^{\ox 2} \ket{aa} \bra{aa} U_i^{\ox 2}}\\
&=|A|\frac{2\Pi^{\sym}_{AA'}}{(|A|+1)|A|} = \frac{\id_{AA'} + F_{AA'}}{|A|+1}\ .\label{eq:next}
\end{align}
We now consider $a_{2}\neq a_{2}'$ and use the fact that the permutations are chosen to be pairwise independent. 
Similar techniques were used in the context of decoupling in~\cite{Szehr11}. We have
\begin{align}
\exc{P}{\left(P^{\dagger}\right)^{\ox 2} \ket{a_{1}a_{2}a_{1}a_{2}'} \bra{a_{1}a_{2}a_{1}a_{2}'} P^{\ox 2} } &= \exc{P}{\proj{P^{-1}(a_{1}a_{2})} \ox \proj{P^{-1}(a_{1}a_{2}')}}\\
&= \sum_{a \neq a'} \probc{P}{P^{-1}(a_{1}a_{2}) = a, P^{-1}(a_{1}a_{2}')=a'} \proj{a} \ox \proj{a'}\\
&= \frac{1}{|A|(|A|-1)}\sum_{a \neq a'}\proj{a} \ox \proj{a'}\\
&= \frac{\id_{AA'}}{|A|(|A|-1)} - \frac{1}{|A|(|A|-1)} \sum_a \proj{aa}\ .\label{eq:after}
\end{align}
Going back to Equation~\eqref{eq:next}, we get together with Equation~\eqref{eq:after} that for any $a_2 \neq a_2'$,
\begin{align}
\exc{P,i}{\left(U_i^{\dagger}\right)^{\otimes 2} \left(P^{\dagger}\right)^{\ox 2} \ket{a_{1}a_{2}a_{1}a_{2}'} \bra{a_{1}a_{2}a_{1}a_{2}'} P^{\ox 2} U_i^{\ox 2} } &=
\frac{\id_{AA'}}{|A|(|A|-1)}-\frac{1}{|A|(|A|-1)}\sum_a \exc{i}{\left(U_i^{\dagger}\right)^{\otimes 2} \proj{aa} U_i^{\otimes 2}} \\
&= \frac{\id_{AA'}}{|A|\left(|A|-1\right)}-\frac{\id_{AA'} + F_{AA'}}{|A|(|A|-1)(|A|+1)} \\
&= \frac{|A|\id_{AA'} - F_{AA'}}{|A|(|A|^{2}-1)}\ .
\end{align}
This being true for all $a_{1},a_{2},a_{2}'$, it follows with Equation~\eqref{eq:start} that,
\begin{align}
&\exc{P,i}{\left\| \cT\left(PU_i \rho_{AE} \left(PU_i\right)^{\dagger}\right) - \frac{\id_{A_{1}}}{A_{1}} \ox \rho_E \right\|_1}\\
&=\sqrt{|A_{1}|\tr\left[ \tilde{\rho}_{AE}^{\ox 2} \left( \frac{\id_{AA'} + F_{AA'}}{|A|+1} + |A|(|A_{2}|-1)\frac{|A|\id_{AA'}-F_{AA'}}{|A|(|A|^{2}-1)}\right)\ox F_{EE'}\right]-\tr\left[\tilde{\rho}_{E}^{2}\right]}\\
&=\sqrt{|A_{1}|\left( \frac{1}{|A| + 1} + \frac{|A|(|A_{2}|-1)}{|A|^{2}-1} \right)\tr\left[\tilde{\rho}_{AE}^{\ox 2} \left(\id_{AA'} \ox F_{EE'}\right) \right]+|A_{1}|\left(\frac{1}{|A| + 1} - \frac{|A_{2}|-1}{|A|^2 - 1} \right) \tr \left[ \tilde{\rho}_{AE}^{\ox 2}\left(F_{AA'} \ox F_{EE'}\right)\right]-\tr\left[\tilde{\rho}_{E}^{2}\right]}\\
&=\sqrt{|A_{1}|\frac{|A||A_{2}| - 1}{|A|^2 - 1}\tr\left[ \tr_{AA'} \left[\tilde{\rho}^{\ox 2}_{AE}\left(\id_{AA'} \ox F_{EE'}\right)\right]\right]+|A_{1}|\left(\frac{|A|-|A_{2}|}{|A|^2 - 1} \right)\tr[ \tilde{\rho}_{AE}^2 ]-\tr\left[\tilde{\rho}_{E}^{2}\right]}\\
&=\sqrt{\left(\frac{|A|^{2}-|A_{1}|}{|A|^2 - 1}-1\right)\tr[ \tilde{\rho}_E^2 ] + \left(\frac{|A_{1}||A|-|A|}{|A|^2 - 1} \right)\tr[ \tilde{\rho}_{AE}^2 ]}\leq\sqrt{\frac{|A_{1}|}{|A|+1}\tr\left[\tilde{\rho}_{AE}^2\right]}=\sqrt{\frac{|A_{1}|}{|A|+1} 2^{-H_{2}(A|E)_{\rho|\sigma}}}\ ,
\end{align}
where we used the definition of the conditional collision entropy (Equation~\eqref{eq:coll}) in the last step. Now, by choosing $\sigma_{E}$ appropriately, and an analogue argumentation as at the very end of the proof of Theorem~\ref{thm:small-decoupling-set}, we conclude that,
\begin{align}
\exc{P,i}{\left\| \cT\left(PU_i \rho_{AE} \left(PU_i\right)^{\dagger}\right) - \frac{\id_{A_{1}}}{A_{1}} \ox \rho_E \right\|_1}
\leq \sqrt{\frac{|A_{1}|}{|A|+1} 2^{-H^{\delta}_{\min}(A|E)_{\rho}}}+2\delta\ .
\end{align}
\end{proof}

\newtheorem*{thm-singleQuditExtract}{Theorem \ref{thm:singleQuditExtract}}
\begin{thm-singleQuditExtract}
Let $A=A_{1}A_{2}$ with $|A|=d^{n}$, $|A_{1}|=d^{\xi n}$, $|A_{2}|=d^{(1-\xi)n}$, and $d$ a prime power. Consider the map $\cT_{A\rightarrow A_{1}}$ as defined in Equation~\eqref{eq:meas-map}. Then for $\delta\geq0$ and $\delta'>0$,
\begin{align}
\frac{1}{|\cP|} \frac{1}{(d+1)^n} \sum_{P \in \cP}\sum_{V \in \cV_{d,n}} &\left\| \cT_{A\rightarrow A_{1}}\left(PV \rho_{AE} \left(PV\right)^{\dagger}\right)- \frac{\id_{A_{1}}}{|A_{1}|} \ox \rho_E \right\|_1\\
&\leq \sqrt{2^{\left(1-\log (d+1)+\xi\log d\right)n}(1+2^{-H^{\delta}_{\min}(A|E)_{\rho}+z})}+2(\delta+\delta')\ ,
\end{align}
where $\cV_{d,n}$ is defined as above, $\cP$ is a set of pair-wise independent permutation matrices, and $z=\log\left(\frac{2}{\delta'^2}+\frac{1}{1-\delta}\right)$. In particular, the set $\{ PV : P \in \cP, V \in \cV_{d,n}\}$ is a $(k, \eps)$-extractor provided
\begin{align}
\log |A_1| \leq (\log(d+1) - 1) n + \min \left\{0, k\right\}  - 4 \log(1/\eps) - 7\, 
\end{align}
and the number of unitaries is
\begin{align}
L = (d+1)^n d^n (d^n - 1)\ .
\end{align}
\end{thm-singleQuditExtract}

\begin{proof}
We use the same strategy as in the proofs of Theorem~\ref{thm:small-decoupling-set} and Theorem~\ref{thm:full-set-mub}; here again with $\tilde{\rho}_{AE}=\left(\id_{A}\otimes\rho_{E}\right)^{-1/4}\rho_{AE}\left(\id_{A}\otimes\rho_{E}\right)^{-1/4}$. We get
\begin{align}
&\frac{1}{|\cP|} \frac{1}{(d+1)^n} \sum_{P \in \cP, V \in \cV_{d,n}} \left\| \cT\left(PV \rho_{AE} \left(PV\right)^{\dagger}\right) - \frac{\id_{A_{1}}}{|A_{1}|} \ox \rho_E \right\|_1
\equiv\exc{P,V}{ \left\| \cT\left(PV \rho_{AE} \left(PV\right)^{\dagger}\right) - \frac{\id_{A_{1}}}{|A_{1}|} \ox \rho_E \right\|_1} \notag \\
&\leq\sqrt{|A_{1}|\tr\left[\tilde{\rho}_{AE}^{\otimes 2}\left(\sum_{a}\exc{V}{\left(V^{\dagger}\proj{a}V\right)^{\otimes2}}+|A|\left(|A_{2}|-1\right)\frac{\id_{AA'}}{|A|\left(|A|-1\right)}\right)\otimes F_{EE'}\right]-\tr\left[\tilde{\rho}_{E}^{2}\right]}\ . \label{eq:singleq1}
\end{align}
We calculate
\begin{align}
\exc{V}{\left(V^{\dagger}\proj{a}V\right)^{\otimes2}}= \frac{1}{(d+1)^n} \sum_{a_1, a_2, \dots, a_n} \sum_{V_1, \dots, V_n} \bigotimes_i \left( ( V^{\dagger}_i \proj{a_i} V_i)^{\otimes 2} \right)= \frac{1}{(d+1)^n} \bigotimes_i \left(\sum_{a_i, V_i} V_i^{\dagger} \proj{a_i} V_i \right)^{\otimes 2}\ . \label{eq:singleq2}
\end{align}
As $\left\{V_0,\dots,V_d\right\}$ form a maximal set of mutually unbiased bases in dimension $d$, and with this form a complex projective 2-design (Lemma~\ref{lem:mub-design}), we have
\begin{align}
\sum_{a \in \{0,\ldots,d\}, V \in \cV_{d,1} } \left(V^{\dagger} \proj{a} V \right)^{\otimes 2} = 2 \Pi^{\sym}\ .
\end{align}
Furthermore $(\Pi^{\sym}_{B})^{\otimes n} \leq \Pi^{\sym}_{B^{\ox n}}$ for any quantum system $B$, and hence we obtain
\begin{align}
\frac{1}{(d+1)^n} \bigotimes_i \left(\sum_{a_i, V_i} V_i^{\dagger} \proj{a_i} V_i \right)^{\otimes 2}
\leq \left(\frac{2}{d+1}\right)^n \Pi^{\sym}_{AA'} = \left(\frac{2}{d+1}\right)^n \frac{\id_{AA'} + F_{AA'} }{2}\ .\label{eq:Mdef}
\end{align}
Together with Equation~\eqref{eq:singleq1} and Equation~\eqref{eq:singleq2}, we get
\begin{align}
&\exc{P,V}{ \left\| \cT\left(PV \rho_{AE} \left(PV\right)^{\dagger}\right) - \frac{\id_{A_{1}}}{|A_{1}|} \ox \rho_E \right\|_1}\\
&\leq\sqrt{|A_{1}|\tr\left[\tilde{\rho}_{AE}^{\ox 2}\left(\left(\frac{2}{d+1}\right)^n\frac{\id_{AA'}+F_{AA'}}{2}+|A|(|A_{2}|-1)\frac{\id_{AA'} }{|A|(|A|-1)}\right)\ox F_{EE'} \right]-\tr\left[\tilde{\rho}_{E}^{2}\right]}\label{eq:urStep}\\
&=\sqrt{\left(\frac{|A|-|A_{1}|}{|A|-1}+\frac{|A_{1}|}{2}\left(\frac{2}{d+1}\right)^n\right)\tr\left[\tilde{\rho}_{AE}^{\ox 2}\left(\id_{AA'} \ox F_{EE'}\right)\right]+\frac{|A_{1}|}{2}\left(\frac{2}{d+1}\right)^n\tr\left[\tilde{\rho}_{AE}^{\ox 2}\left(F_{AA'}\ox F_{EE'}\right)\right]-\tr\left[\tilde{\rho}_{E}^{2}\right]}\\
&\leq\sqrt{\left(1+2^{(1-\log(d+1)+\xi\log d)n}\right)\tr\left[\tr_{AA'}\left[\tilde{\rho}^{\ox 2}_{AE}\left(\id_{AA'}\ox F_E\right)\right]\right]+2^{(1-\log(d+1)+\xi\log d)n}\tr[\tilde{\rho}_{AE}^2]-\tr\left[\tilde{\rho}_{E}^{2}\right]}\\
&=\sqrt{2^{(1-\log(d+1)+\xi\log d)n}\tr\left[\tilde{\rho}_E^2\right]+2^{(1-\log(d+1)+\xi\log d)n}\tr\left[\tilde{\rho}_{AE}^2\right]}\\
&=\sqrt{2^{(1-\log(d+1)+\xi\log d)n}\left(1+2^{-H_{2}(A|E)_{\rho|\rho}}\right)}\ ,
\end{align}
where we used the definition of the conditional collision entropy (Equation~\eqref{eq:coll}) in the last step. Now, by an analogue argumentation as at the very end of the proof of Theorem~\ref{thm:small-decoupling-set}, we conclude that,
\begin{align}
\exc{P,V}{ \left\| \cT\left(PV \rho_{AE} \left(PV\right)^{\dagger}\right) - \frac{\id_{A_{1}}}{|A_{1}|} \ox \rho_E \right\|_1}
\leq\sqrt{2^{\left(1-\log (d+1)+\xi\log d\right)n}(1+2^{-H^{\delta}_{\min}(A|E)_{\rho} + z})}+2(\delta+\delta')\ .
\end{align}
Setting $\delta = 0$ and $\delta' = \eps/4$, we conclude that the set $\{PV : P \in \cP, V \in \cV_{d,n}\}$ is a $(k, \eps)$-QC-extractor provided 
\begin{align}
\log |A_1| = n \cdot \xi \log d &\leq (\log(d+1) - 1) n - \log(1 + 2^{-k + \log(8/\eps^2 + 1)}) + \log((\eps/2)^2)  \\
&\leq (\log(d+1) - 1) n + \min \left\{0, k - \log(8/\eps^2 + 1)\right\} - 1 - 2\log(1/\eps) - 2 \\
&\leq (\log(d+1) - 1) n + \min \left\{0, k\right\}  - 4 \log(1/\eps) - 7.
\end{align}
\end{proof}

Note that step~\eqref{eq:urStep} is indeed striking when we consider the case of trivial side information. Effectively, one of the terms we wish to bound
then is $\tr\left[\rho_A^{\otimes 2} M\right]$ where $M$ is given by the l.h.s.~of~\eqref{eq:Mdef}. This, however, is exactly what one bounds when proving entropic uncertainty relations for MUBs~\cite{ballester07}, or more generally anti-commuting measurements~\cite{ww:cliffordUR}. And indeed, in the case with quantum side information, our techniques also allow to directly derive entropic uncertainty relations with quantum side information in terms of the quantum conditional collision entropy (as defined in Equation~\eqref{eq:coll}) using the fact that MUBs form a complex projective 2-design. However, we are more interested in relations in terms of the min-entropy (see Section~\ref{sec:URbounds}). On the other hand, it is an interesting question whether the techniques from~\cite{ww:cliffordUR} can be extended to give a better bound than the (probably too general) eigenvalue bound of~\eqref{eq:Mdef}.


\section{Proofs of Uncertainty Relations}\label{sec:urProofs}

In this section, we provide the full proofs regarding our claims of entropic uncertainty relations.

\newtheorem*{thm:neumann}{Proposition \ref{thm:neumann}}
\begin{thm:neumann}
Let $d \geq 2$ be a prime power, and $\left\{V_0,V_1,\dots,V_d\right\}$ define a complete set of MUBs of $\CC^d$. Consider the set of measurements $\{\cM^j_{A \to K} : j \in [(d+1)^n]\}$ on the $n$ qudit space $A$ defined by the unitary transformations $\left\{V=V_{u_1}\ox\cdots\ox V_{u_n}|u_i\in\left\{0,\dots,d\right\}\right\}$. Then for all $\rho_{AE} \in \cS(AE)$, we have
\begin{align}
\frac{1}{(d+1)^n}\sum_{j=1}^{(d+1)^n}H(K|E)_{\rho^{j}}\geq n\cdot\left(\log(d+1)-1\right)+\min\left\{0,H(A|E)_{\rho}\right\}\ ,
\end{align}
where $\rho^j = \cM^j_{A \to K}(\rho)$.
\end{thm:neumann}

\begin{proof}
Using the QC-extractor for the single-qudit MUB case as discussed in Section~\ref{sec:singlequdit}, we get with the same reasoning as before that for $\eps>0$, $\delta\geq0$,
\begin{align}
	\frac{1}{L}\sum_{j=1}^{L} H(K|E)_{\rho^{j}} &\geq(1-\eps)\left(n\left(\log(d+1)-1\right)-\log\left(1+2^{-\hmin^{\delta}(A|E)_{\rho|\rho}}\right)-\log\left(\frac{1}{(\eps-2\delta)^{2}}\right)\right)-2h(\eps) \notag \\
	&\geq(1-\eps)\left(n\left(\log(d+1)-1\right)+\min\left\{0,\hmin^{\delta}(A|E)_{\rho|\rho}\right\}-1-\log\left(\frac{1}{(\eps-2\delta)^{2}}\right)\right)-2h(\eps) \label{eq:aepqudit}\ .
\end{align}
Here we use a version with $\hmin^{\delta}(A|E)_{\rho|\rho}$ instead of $\hmin^{\delta}(A|E)_{\rho}$, but this is immediate from the proof of Theorem~\ref{thm:singleQuditExtract}. Evaluating Equation~\eqref{eq:aepqudit} on the $m$-fold tensor product of the original input system $d^{n}$, and multiplying both sides with $1/m$, we obtain
\begin{align}
\frac{1}{L}\sum_{j=1}^{L} H(K|E)_{\rho^{j}}&\geq(1-\eps)\left(n\left(\log(d+1)-1\right)+\min\left\{0,\frac{1}{m}\hmin^{\delta}(A|E)_{\rho^{\otimes m}|\rho^{\otimes m}}\right\}\right)\\
&-\frac{1-\eps}{m}\left(1-\log\left(\frac{1}{(\eps-2\delta)^{2}}\right)\right)-\frac{2h(\eps)}{m}\\
&\geq(1-\eps)\left(n\left(\log(d+1)-1\right)+\min\left\{0,H(A|E)_{\rho}-\frac{4\sqrt{1-2\log\delta}\left(2+\frac{n}{2}\right)}{\sqrt{m}}\right\}\right)\\
&-\frac{1-\eps}{m}\left(1-\log\left(\frac{1}{(\eps-2\delta)^{2}}\right)\right)-\frac{2h(\eps)}{m}\ .
\end{align}
Here we used the fully quantum asymptotic equipartition property for the smooth conditional min-entropy (Lemma~\ref{lem:aep}). By first letting $m\rightarrow\infty$ and then $\eps\rightarrow0$, we arrive at the claim.
\end{proof}


\section{Definition Weak String Erasure}

For convenience sake, we here provide a formal definition of weak string erasure~\cite{noisy:new} for $p \neq 1/2$ as given in~\cite{prabha:limits}, where we restrict to qubits ($d=2$).
The definition is stated in terms of ideal states, akin to an ideal functionality in classical cryptography. In the proof of security against dishonest Bob, we
simply show that Bob's $\eps$-\emph{smooth} min-entropy is high. However, by~\eqref{eq:purifiedVStrace} this implies that Bob's real state is $\eps$-close to an ideal state of high min-entropy in trace distance.
Note that for cryptographic purposes, we will specify distances in term of the \emph{trace} distance, since this is the relevant distance that determines how well the real protocol can be distingiushed from
the ideal state~\cite{helstrom}.

In the definition below, we will need to talk about distributions over subsets $\cI \subseteq [n]$, where each element of $[n]$ has probability $p$ of being in $\cI$. 
Clearly, the probability that Bob learns a particular subset $\cI$ satisfies
\begin{align}\label{eq:probDist}
	\Pr(\cI) = p^{|\cI|} (1-p)^{n-|\cI|}
\end{align}
Note that we can write the subset $\cI$ as a string $(y_1,\ldots,y_n) \in \{0,1\}^n$ where $y_i = 1$ if and only if
$i \in \cI$, allowing us to identify $\ket{\cI}= \ket{y_1} \otimes \ldots \otimes \ket{y_n}$. The probability
distribution over subsets $\cI \subseteq [n]$ can then be expressed as (see also~\cite{noisy:new})
\begin{align}
	\Psi(p)= \sum_{\cI \subseteq 2^{[n]}} p^{|\cI|} (1-p)^{n-|\cI|} \proj{\cI}\ .
\end{align}
Furthermore, we will follow the notation of~\cite{noisy:new} and use
\begin{align}
	\tau_{\cS}= \frac{1}{|\cS|} \sum_{s \in \cS} \proj{s}\ ,
\end{align}
to denote the uniform distribution over a set $\cS$. 

\begin{definition}[\textbf{Non-uniform WSE}]\label{def:wse}
An $(n,\lambda, \varepsilon,p)$-weak string erasure scheme is a protocol between A and B satisfying the following properties:

\textbf{Correctness:} If both parties are honest, then there exists an ideal state $\sigma_{X^{n}\mathcal{I}X_{\mathcal{I}}}$ such that
\begin{enumerate}
\item The joint distribution of the $n$-bit string $X^{n}$ and subset $\mathcal{I}$ is given by
\begin{equation}\label{eq:wsecorrect}
\sigma_{X^{n}\cI} = \tau_{\{0,1\}^{n}}\otimes \Psi(p)\ ,
\end{equation}
\item The joint state $\rho_{AB}$ created by the real protocol is equal to the ideal state: $\rho_{AB} = \sigma_{X^{n}\cI X_{\cI}}$ where we identify $(A,B)$ with $(X^{n},\cI X_{\cI})$.
\end{enumerate}

\textbf{Security for Alice:} If A is honest, then there exists an ideal state $\sigma_{X^{n}B'}$ such that 
\begin{enumerate}
\item The amount of information $B'$ gives Bob about $X^{n}$ is limited:
\begin{equation}\label{eq:honestA}
\frac{1}{n}\hmin(X^{n}|B')_{\sigma} \geq \lambda
\end{equation}
\item The joint state $\rho_{AB'}$ created by the real protocol is $\eps$-close to the ideal state in trace distance, where we identify $(X^{n},B')$ with $(A,B')$.
\end{enumerate}

\textbf{Security for Bob:} If B is honest, then there exists an ideal state $\sigma_{A'\hat{X}^{n}\cI}$ where $\hat{X}^{n} \in \{0,1\}^{n}$ and $\cI \subseteq [n]$ such that
\begin{enumerate}
	\item The random variable $\cI$ is independent of $A'\hat{X}^{n}$ and distributed over $2^{[n]}$ according to the probability distribution given by~\eqref{eq:probDist}:
\begin{equation}\label{eq:honestB}
\sigma_{A'\hat{X}^{n}\cI} = \sigma_{A'\hat{X}^{n}} \otimes \Psi(p)\ .
\end{equation}
\item The joint state $\rho_{A'B}$ created by the real protocol is equal to the ideal state: $\rho_{A'B} = \sigma_{A'(\cI\hat{X}_{\cI})}$, where we identify $(A',B)$ with $(A',\cI\hat{X}_{\cI})$.
\end{enumerate}
\end{definition}



\end{document}